\documentclass[10pt]{article}

\usepackage{IEEEtrantools}
\usepackage{amssymb,amsthm}
\usepackage[pdftex,dvips]{graphicx}
\usepackage{epsfig}
\graphicspath{{figs/}}
\DeclareGraphicsExtensions{.pdf,.eps}
\usepackage{fullpage}
\usepackage{cite}
\usepackage[cmex10]{amsmath}
\usepackage{algorithm}
\usepackage[noend]{algorithmic}
\usepackage{array}
\usepackage{mdwmath}
\usepackage{mdwtab}
\usepackage{eqparbox}
\usepackage[caption=false,font=footnotesize]{subfig}
\usepackage{fixltx2e}
\usepackage{url}
\usepackage{wrapfig}
\usepackage{floatflt}
\usepackage{picins}
\usepackage[]{latexsym}

\newtheorem{defn}{\bf Definition}
\newtheorem{lemma}{Lemma}
\newtheorem{thm}{Theorem}
\newtheorem{fact}{\bf Fact}
\newtheorem{rem}{Remark}

\newcommand{\G} {{\mathcal{G}}}

\renewcommand{\L} {{\mathcal{L}}}

\renewcommand{\P} {{\mathcal{P}}}

\newcommand{\V} {{\mathcal{V}}}
\newcommand{\E} {{\mathcal{E}}}
\newcommand{\C} {{\mathcal{C}}}

\newcommand{\T} {{\mathcal{T}}}
\newcommand{\M} {{\mathcal{M}}}

\newcommand\comment[1]{}

\newcommand{\appsection}[1]{\let\oldthesection\thesection
  \renewcommand{\thesection}{Appendix \oldthesection}
  \section{#1}\let\thesection\oldthesection}

\DeclareMathOperator*{\argmin}{arg\,min}

\makeatother

\def\algorithmicrequire{\textbf{Inputs:}}
\def\algorithmicensure{\textbf{Output:}}

\title{On Channel-Discontinuity-Constraint Routing in Wireless Networks}

\author{Swaminathan Sankararaman\thanks{Department of Computer Science, University of Arizona, {\sf swami@email.arizona.edu}.}
\and
Alon Efrat \thanks{Department of Computer Science, University of Arizona, {\sf alon@cs.arizona.edu}.} 
\and
Srinivasan Ramasubramanian \thanks{Department of Electrical and Computer Engineering, University of Arizona {\sf srini@ece.arizona.edu}.}
\and
Pankaj K. Agarwal \thanks{Department of Computer Science, Duke University {\sf pankaj@cs.duke.edu}.}
}
\date{}

\begin{document}

\maketitle

\begin{abstract}

Multi-channel wireless networks are increasingly being
employed as infrastructure networks, e.g.\ in metro areas.
Nodes in these networks frequently  employ directional antennas to improve
spatial throughput. In such networks, given a source and destination, it is of
interest to compute an optimal path and channel assignment on every link in the
path such that the path bandwidth is the same as that of the link bandwidth and
such a path satisfies the constraint that no two consecutive links on the path
are assigned the same channel, referred to as ``Channel Discontinuity
Constraint'' (CDC). CDC-paths are also quite useful for TDMA system,
where preferably every consecutive links along a path are assigned different
time slots.

This paper contains several contributions.  We first present an
$O(N^{2})$ distributed algorithm for discovering the shortest CDC-path between
given source and destination. For use in wireless networks, we explain how
spatial properties can be used for dramatically expedite the algorithm. This
improves the running time of the $O(N^{3})$ centralized algorithm  of Ahuja et
al. for finding  the minimum-weight CDC-path. Our second result is a generalized
$t$-spanner for CDC-path; For any $\theta>0$  we show how to construct a
sub-network containing only $O(\frac{N}{\theta})$ edges, such that that length
of shortest CDC-paths between arbitrary sources and destinations increases by
only a factor of at most $(1-2\sin{\tfrac{\theta}{2}})^{-2}$. We propose a
novel algorithm to compute the spanner in a distributed manner using only
$O(n\log{n})$ messages. This scheme can be implemented in a distributed manner
using the ideas of \cite{Arango2009} with a message complexity of $O(n\log{n})$
and it is highly dynamic, so addition/deletion of nodes are easily handled in a
distributed manner. An important conclusion of this scheme is in the case of
directional antennas are used. In this case, it is enough to consider only the
two closest nodes in each cone.

\end{abstract}

\section{Introduction}

Wireless infrastructure networks (WINs) are gaining prominence as they
are being increasingly deployed in metro
areas to provide ubiquitous information access~\cite{survey}. WINs
provide a low-cost scalable network, support broadband data, and allow
use of unlicensed spectrum. WINs have a wide area of applications,
including public internet access \cite{survey}, PORTAL \cite{Portal},
video streaming \cite{video_streaming}, and underground mining \cite{mining}.
In order to increase the bandwidth in the WINs, nodes employ multiple
wireless transceivers (interface cards) to achieve simultaneous
transmission/reception over multiple orthogonal channels.

Recent research has focused on effectively harvesting the available
bandwidth in a wireless network. The wireless interference constraint
is the key factor that limits the achievable throughput.  The interference
is encountered in two ways: (1) a node may not receive from
two different nodes on the same channel at any given time; and (2)
a node may not receive and transmit on the same channel at any given
time. Moreover, if omnidirectional antenna is employed, then there may
be at most only one node transmitting on a channel in the vicinity of a
node that is receiving on that channel. The interference constraints in wireless
networks may be divided into two categories~\cite{BROADNETS08}:
\textit{inter-flow} and \textit{intra-flow}. The inter-flow
interference refers to the scenario where two links belonging to different
flows cannot be active (on the same channel) at the same time as one
receiver will experience interference due to the other transmission.
The intra-flow interference refers to the scenario where two links
belonging to the same flow cannot be active (on the same channel)
at the same time. The same problem arises also in the TDMA setting,
where a node can be used for streaming applications, but has to
receive and transmit messages in different time slots.

\textbf{Prior Work}: The problem of routing and channel assignment in WINs
refers to computing paths and channel assignment on the paths such that there
are no inter-flow and intra-flow interferences. If the bandwidth of a link is
$B$, then the end-to-end throughput on the path is also $B$ as all the
links in the path can be active simultaneously. The problem of joint
routing and channel assignment is hard when nodes employ omnidirectional
antennas, hence is typically solved as two independent sub-problems.
For a given set of calls where routing is known, the problem of channel
assignment may be mapped to distance-2 vertex and edge coloring
problems~\cite{distance-2}. Using such a mapping, the objective is to compute
channel assignment satisfying the limit on the number of transceivers at each
node. Distance-2 vertex and edge coloring problems are well known NP-hard
problems, hence the problem of channel assignment for networks employing
omnidirectional antennas. Various approximation algorithms and heuristics have
been developed for distance-2 coloring with different objectives, such
as minimizing interference, maximizing throughput, and minimizing
the number of required
channels~\cite{vertex_1,vertex_2,vertex_3,vertex_4,edge_1,edge_2,edge_3} .
An approach based on balanced incomplete block design is developed
in~\cite{BIBD} to assign channels for each interface card such that
the communication network is 2-edge-connected with minimum interference.
In \cite{SAFE}, a heuristic based on random channel assignment policy
is developed to maintain connectivity of the network.

Among the works that compute paths for multi-channel networks, \cite{Draves}
develops a routing protocol to find an efficient path with low intra-flow
interference, by taking into account link loss rate, link data rate,
and channel diversity. In \cite{iAWARE}, shortest path with low interference
is computed based on an extension to the AODV protocol to account
for (inter- and intra-flow) interference and link data and loss rates.
In the space of wireless network design for a given static traffic,
centralized and distributed approaches for joint channel assignment
and routing in the multi-interface WMN with the objective of maximizing
throughput are developed in \cite{Raniwala1,Raniwala2,Bhatia,Meng},
while \cite{Kodialam} considers the objective of achieving a given
data rate.

It has been shown in~\cite{capacity_improve,capacity_improvement,capacity_bound}
that the capacity of WINs may be further improved by increasing spatial
reuse by employing directional antennas. The problem of channel assignment
in networks employing directional antennas may be mapped to the edge-coloring
problem~\cite{Ramanathan2}. In \cite{DMesh}, a network architecture
with nodes employing non-steerable directional antennas is developed.
The authors develop approaches for routing and channel assignment
by considering tree-based topologies rooted at ``gateway'' nodes.

There are indications that the problem of finding a path and channel
assignment such that all links can be active simultaneously is NP-complete when
nodes employ omnidirectional antennas. Our algorithms are designed for the
frequent cases for which the effect of interference can be
ignored. For example, when directional antennas are used, as their
prices and accuracy are rapidly improving, the nodes may be carefully placed
such that any two independent links can use the same channel.

The first work addressing joint routing and channel assignment in
WINs under such scenarios is the recent paper by Ahuja et. al.
\cite{BROADNETS08}. In this case, the path bandwidth will be the same as an
individual link bandwidth if no two consecutive links on the path are assigned
the same channel. We refer to the constraint on channel assignment as channel
discontinuity constraint (CDC) and any path that satisfies the constraint as a
CDC-path. In graph theory literature, CDC-paths are referred to as {}``properly
edge-colored'' paths, where channels correspond to colors. From now
on, whenever it is mentioned that a CDC-path is found, it is not
mentioned explicitly but it is implied that a channel assignment is
also found.

We can overcome, to some extent, the difficulties that interference
causes in the case of omnidirectional antennas using the known
technique of {\em network coding}. Consider a streaming application, where node
$s_{i}$ forwards the message $m_{j}$ to node $s_{i+1} $, and receives the
message $m_{j+1}$ at the same time from $s_{i-1}$. Meanwhile, node $s_{i+1}$
forwards a message $m_{j-1}$ received earlier from $s_{i}$ to $s_{i+2}$. In the
case where omnidirectional antennas are used, and we have applied the
CDC-path protocol, it is quite possible that the frequency of
transmission of both $s_{i-1}$ and $s_{i+1}$ is $f_{1}$, as it is
different than the frequency $f_2$ used by $s_{i}$ (see Figure
\ref{fig:stream}).  However, as demonstrated in the paper by
\cite{pnc}, since  $s_{i}$ ``knows'' $m_{j-1}$ it can subtract it from
the combined signal $m_{j-1} + m_{j+1}$. See \cite{pnc} for details.
However, this is not a trivial hardware modification, so we would
concentrate in this paper mostly with other applications where CDC
path are useful.

\begin{figure}[t]
\centering
\includegraphics[width=0.9\columnwidth]{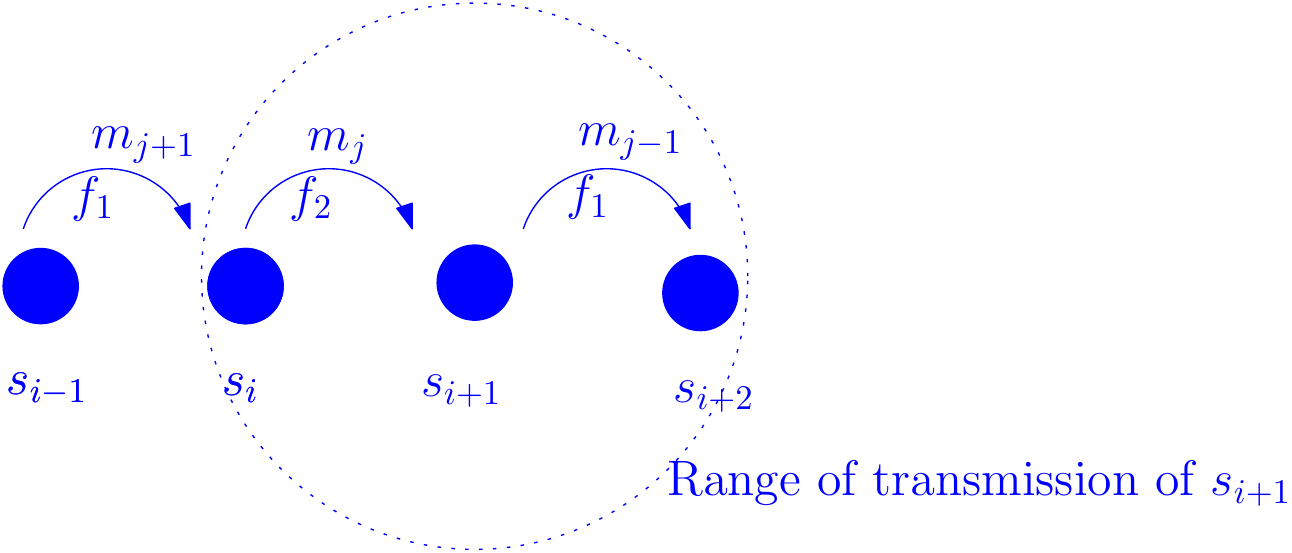}

\caption{A streaming application where node $s_{i}$ forwards message
$m_{j}$ to node $s_{i+1}$ using frequency $f_{2}$ and receives message
$m_{j+1}$ at the same time from $s_{i-1}$ using frequency $f_{1}$.
Node $s_{i+1}$ uses frequency $f_{1}$ to send message $m_{j-1}$
received earlier to $s_{i+2}$. This causes interference at $s_{i}$.}

\label{fig:stream}
\end{figure}

\textbf{Motivation:} As stated, CDC-paths are particularly useful in streaming
applications where the transmissions are continuous and any period of inactivity
reduces bandwidth. CDC-paths are also quite useful in the context of a TDM or
FDM/TDM system, where the channel may be thought of as a time slot or a
combination of frequency and time slot. In TDM systems, it is desirable to
assign different time slots to adjacent transmissions along the path and this
corresponds exactly to the Channel Discontinuity Constraint where each channel
is a time slot. Thus, if $(u,v)$ and $(v,w)$ are adjacent links along the path,
the throughput is maximized when they are using different time slots. Now,
because of other ``conversations'' which have ended, certain time slots may
become available and finding new CDC-paths may help maintain the maximum
bandwidth.

\textbf{Our Contributions:} We present a distributed algorithm for finding the
shortest CDC-path between two nodes in a graph. The algorithm requires the
exchange of $O(N^{2})$ fixed-size messages in total ($O(N)$ per node) where $N$
is the number of nodes. This improves the previous algorithm presented in
\cite{BROADNETS08} which is centralized, and requires $O(N^{3})$ running time.

The second contribution is the construction of a sparse graph $\G'$ which is a
$t$-Spanner for the problem of CDC routing, where $t$ is a controlled parameter.
The number of links in $\G'$ is $O(N)$. By spanner, we mean that for any fixed
$t$ as close as one wishes to $1$, and for any pair of nodes $u$ and $v$ in the
network $\G$, if $\G$ contains a route between then so does the spanner $\G'$,
and its length is longer than the original route by only a factor of $t$.  Our
spanner is based on the Yao graph\cite{Yao}. The basic idea is that at a node,
we divide the region around the node into sectors and store only a few neighbors
in each sector. The parameter $t$ is dependent on the number of sectors. This is
particularly applicable when using directional antennas. For example, if the
antennas' covering range can be abstracted as a sector of angle $30$ degrees,
then by storing for every node just the few neighbors in each sector, we retain
connectivity and gain routes which are no longer than $4.3$ times the
theoretical bounds. This spanner is also highly dynamic, so insertion/deletion
of links is carried out easily.

The rest of the paper is organized as follows. In Section \ref{sec:System
Model}, we discuss the network model and interference constraints. In Section
\ref{sec:prob-form}, we review the technique developed in \cite{BROADNETS08}
that explains how the CDC problem can be expressed as a matching problem. In
Section \ref{altpath}, we present a distributed algorithm for finding a single
CDC-path between $s$ and $d$ which requires the sending of a total of $O(N^{2})$
messages. In Section \ref{spanner}, we discuss a t-spanner for WINs containing
$O(N)$ links. Finally, we present the conclusion and some future work.

\section{System Model}
\label{sec:System Model}
Consider a multi-channel wireless network, where ${\C}$ denotes the set of
orthogonal channels each with bandwidth $B$. Let $C=|{\C|}$. Let ${\V}$ denote
the set of nodes each equipped with $C$ transceivers that may be tuned to any of
the orthogonal channels. Let $R$
be the transmission range of a node. We refer to two nodes as \emph{neighbors}
if the Euclidean distance between them is not greater than $R$. Let $\G(\V,\L)$
denote the connectivity graph, where ${\L}$ denotes the set of links. A link
connects two neighboring nodes.

WINs may be designed such that the number of links that can be active
simultaneously can be maximized. A node $z$ is said to be collinear with a
transmission from $x$ to $y$ if node $z$ cannot receive on the same channel as
that used by the transmission from $x$ to $y$. The interference at node $z$ due
to the transmission from $x$ to $y$ results in signal-to-noise ratio that is
lower than the threshold for decoding the received signal. The collinearity
constraint results in the dependence on channel assignment between links that
are not adjacent to each other. (hence resulting in distance-2 coloring
problems).

While the collinearity constraint is inherent in networks employing
omnidirectional antennas, they can be eliminated completely by careful placement
of nodes when directional antennas are employed. In such carefully planned
wireless infrastructure networks, there are only two interference constraints:
(1) a node may not receive on a channel from more than one node at the same
time; and (2) a node may not transmit and receive on a channel at the same time.
These two constraints may lead to intra-flow interference, however it is limited
to only two adjacent links. In order to avoid intra-flow interference, no two
consecutive links in the path are assigned the same channel.

We assume that the list of channels that a node can transmit/receive on is known
at all times. The channels available on a link between neighboring nodes $u$ and
$v$ is simply the set of common channels at the two nodes. Every call is assumed
to have a bandwidth requirement of one channel capacity. We assume that calls
are bidirectional and the same channel will be shared in both directions on a
link. We also do not care if the communications required for the algorithm
satisfy CDC, since, typically, a control channel is used to transmit
network-related information rather than the actual data transmissions in the
network and the control messages are also typically very short.

\section{Problem Formulation}
\label{sec:prob-form}

\noindent
\textbf{Problem Statement:} {
\em
Given a multi-channel wireless network with no collinear interference, the set
of available channels on every link, the cost of the links, and a node pair
$(s,d)$, find the shortest path between $s$ and $d$ along with a channel
assignment on every link of the path such that no two consecutive links in the
path are assigned the same channel.
}

\begin{defn}
We call a path a \textbf{CDC-path} if we can assign channels $c$ to each link
$(u,v)$ along this path and no two adjacent links have the same channel
assigned.
\end{defn}

Ahuja et. al. \cite{BROADNETS08} showed the equivalence of computing the minimum
cost CDC-path to the computation of minimum cost perfect matching (MCPM) using
Edmonds-Szeider (ES) expansion for nodes. For the sake of completeness, we give
a description of their method here. In the ES-expansion graph, we expand each
node $x$, except the source and destination nodes, into $2C+2$ sub-nodes,
denoted $(x_{1},x_{1}'),(x_{2},x_{2}'),\ldots,(x_{C},x_{C}')$, each pair
corresponding to a channel $c\in{\C}$. We also add sub-nodes $(x_{g},x_{g}')$.
Each pair is connected to each other and we refer to these links as {\em channel
link}. Every sub-node $x_{c}'$ is also connected to $x_{g}$ and $x_{g}'$. Assign
cost 0 to all the above links in the expanded node. Ahuja et. al.
\cite{BROADNETS08} also noted that these expansions may be modified to reduce
the number of links in the following manner. For links with two or less channels
available on them, expand as before. For links with three or more channels, the
additional links may be eliminated since, for any CDC-path which uses this link,
we can always find a channel assignment from the three corresponding channels.

If two nodes $x$ and $y$ are connected in ${\G}$ and have a set of channels
available between them, then the sub-nodes $x_{c}$ and $y_{c}$ corresponding to
the available channels are also connected. Connect $s$ and $d$ to the sub-nodes
of its neighboring nodes corresponding to the available channels between them.

Now, the problem of computing the shortest CDC-path is the same as computing the
minimum-cost perfect matching in the expanded graph. A perfect matching in a
graph is a set of non-adjacent edges, i.e., no two edges share the same vertex,
such that all the vertices in the graph are covered. A minimum cost perfect
matching is defined as a perfect matching where the sum of the cost of the edges
in the matching is minimum.

We illustrate the equivalence of MCPM and CDC-path using an example. Consider
the example network and its ES-expansion as shown in Figure
\ref{fig:eg-Edmonds-Szeider expansion}. Numbers over links represent the
available channels on the corresponding links. The link costs are assumed to be
1 for this illustration, hence the objective is to compute the minimum hop
CDC-path. The edges in the matching are shown in bold in the expanded graph. The
nodes not present in the shortest path from $s$ to $d$ would find matching
within itself (see nodes $y$ and $z$). Intermediate nodes (if any) in the
shortest $s$ to $d$ path will have exactly two incident external edges in the
matching (see node $x$). Clearly, the returned path $s$--$x$--$d$ with channels
1 and 3 over links $s$--$x$ and $x$--$d$, respectively, is the shortest CDC-path
between $s$ and $d$ in terms of hop length.

\begin{figure}
\centering
\includegraphics[width=0.75\columnwidth]{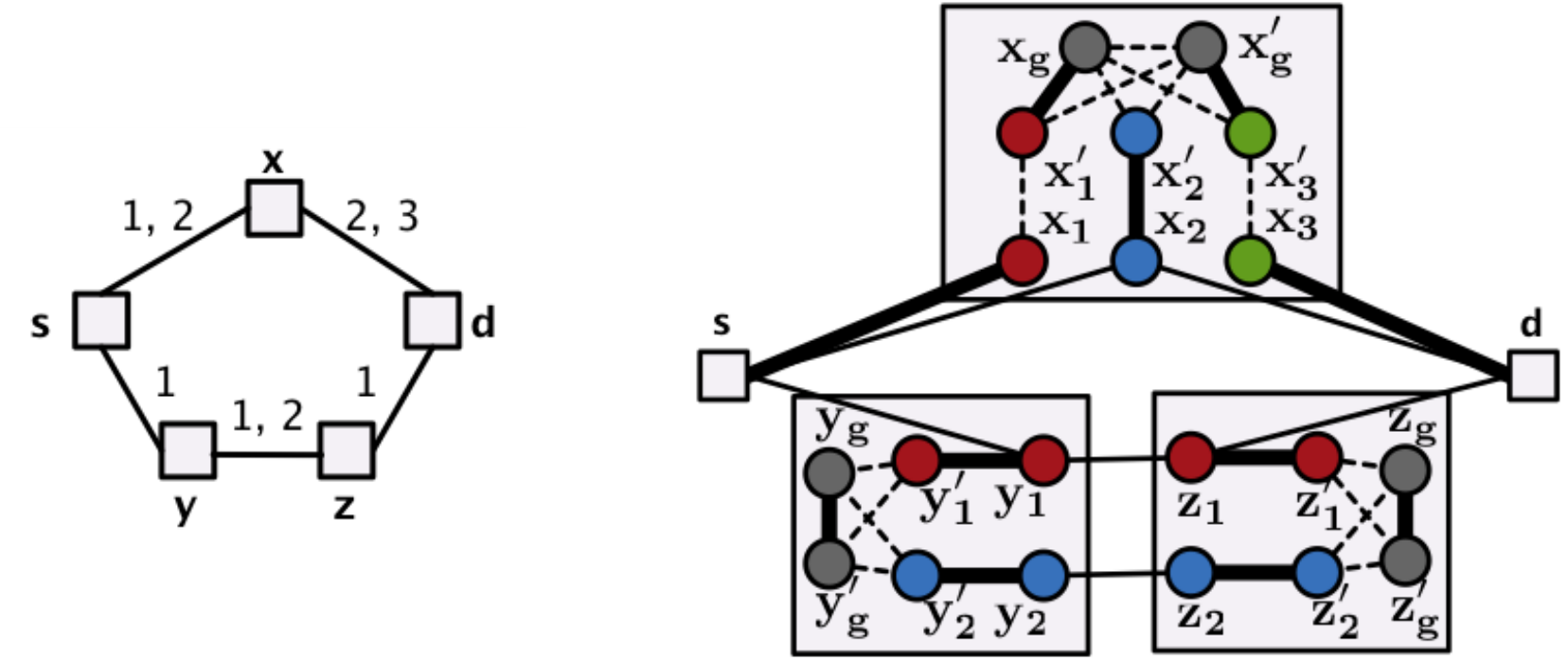}

\caption{Example network with source node $s$ and destination node $d$ and its
ES expansion. The bold links in the expansion graph show the links in the
perfect matching set. Observe that the sequence of external links that are in
the matching provide the CDC-path from $s$ to $d$: $s$--$x$--$d$.}

\label{fig:eg-Edmonds-Szeider expansion}
\end{figure}

The best known sequential implementation of Edmonds' minimum cost perfect
matching algorithm is by Gabow \cite{Gabow} with complexity $O(|\V||\L|)$ where
$|\V|$ and $|\L|$ are the number of nodes and links respectively. Since the
number of nodes in the expanded graph is $O(C|\V|)$ and the number of links is
$O(C(|\L|+|\V|))$, the complexity of computing a solution to MCPM is
$O(|\V||\L|C^{3})$.  However, this complexity can be reduced by observing that
most of the matching computed in the MCPM are within the nodes itself (except
for the links involved in the path).  A closer look reveals that we can start
with an expanded graph with a partial matching where all the internal edges
$x_{c}$ and $x_{c}'$ $\forall c$ at node $x$ and $x_{g}$ and $x'_{g}$ are
matched. If this is the case, then the problem of computing the minimum cost
perfect matching is transformed to computing the minimum cost alternating path.
We thank Kasturi Varadarajan for this observation. It is interesting to note
that Varadarajan and  Agarwal have designed algorithms\cite{kasturi-divconq,kasturi-approx} 
for exactly\cite{kasturi-divconq} and approximately\cite{kasturi-approx} finding a
minimum-weight matching in a geometric setting. In the following sections, we 
develop a distributed algorithm to compute the shortest alternating path with 
better bounds than Ahuja et. al.\cite{BROADNETS08}.

\section{\label{altpath}Finding the Shortest CDC Path}
We are given the expanded graph $\G(\V,\E)$ together with the matching $\M$ of
zero weight (all vertices except $s$ and $d$ internally matched). For any node
$u$, let $MATE(u)$ denote the mate of $u$ in the matching $\M$.

\begin{defn}
Between any two nodes $s$ and $d$, we define an \emph{\textbf{alternating path}}
to be a path with alternating unmatched and matched edges.
\end{defn}

\begin{defn}
For any node $u$, we define a \emph{\textbf{co-link path}} from $s$ to $u$ as an
alternating path $\pi=\{u_{0}=s,u_{1},u_{2},\dots,u_{k}=u\}$ where
$(s,u_{1}),(u_{k-1},u)\notin \M$.
\end{defn}

Our algorithm shares many ideas with Edmonds' algorithm. It works by finding
minimum weight co-link paths similar to each phase of Edmonds' algorithm. Each
node $u$ maintains two distances: $d_{T}[u]$ and $d_{S}[u]$, corresponding to
the minimum weight co-link paths to itself and $MATE(u)$ respectively. Through
the course of the algorithm, each node is given labels from the set $\{S,T,F\}$.
We say a node $u\in S$ if it is given an $S$ label and similarly for the other
labels. Initially, $s\in S$ and every other node $u\in F$. Each node $u$ also
maintains the current known distances $d_{S}$ and $d_{T}$ of its neighbors in
$S$. During the algorithm, certain odd subsets of vertices are termed as
\emph{blossoms}.

\begin{defn}
A \emph{\textbf{blossom}} $B$ is an odd circuit in $G$ for which $M\cup B$ is a
perfect matching for all vertices in $B$ except one. The lone unmatched vertex
is termed as the \emph{\textbf{base}} of the blossom.
\end{defn}

Through the course of the algorithm, a node $u$ may be added to one or more
blossoms. Hence, it maintains the ID of the base of the outermost blossom it
belongs to (the reader is referred to the book by Lawler\cite{Lawler1976} for
the definition of outermost blossom). Let $b[u]$ denote this ID. When a node $u$
is added to $S$ for the first time, it stores the corresponding ID of the parent
of $u$, $P_{S}[u]$. When a node $u$ is added to $T$ for the first time, if it is
not the base of a blossom, it stores the corresponding ID of the parent of $u$,
$P_{T}[u]$. If it is the base of a blossom $B$, $B$ is now a subblossom of a new
blossom $B'$. In this case, $u$ stores the edge $(w,x)$ which is the first
unmatched edge in $B'$ connected to $B$ and not in $B$. For more details, see
\cite{Lawler1976}.

The algorithm works in three phases: FINDMIN, BLOSSOM and GROW. Before we
proceed, let us define a few expressions. For a node $u$, let 
\[
val[v] = \begin{cases}
\frac{d_{S}[u]+d_{S}[v] + w(u,v)}{2} &\text{ if $v\in S$}, \\
d_{S}[u]+w(u,v) &\text{ if $v\in F$}
\end{cases}
\]
Each node $u$ maintains an ``examined'' state for each neighbor $v\in S,F$. Let
$minval[u] = \displaystyle{\min_{v} val[v]}$ and $v_{min} =
\displaystyle{\argmin_{v\text{ not yet examined}} val[v]}$.  The steps of FINDMIN are as follows.

\noindent\textbf{(Step $1$)} Each node $u$ computes $minval[u]$ and the
corresponding $v_{min}$. 

\noindent\textbf{(Step $2$)} If $v_{min}\in S$, it sends a message to $v_{min}$
requesting $b[v_{min}]$. If $b[v_{min}]=b[u]$, the two nodes belong to the same
blossom and no blossom discovery is necessary. In this case, $u$ marks $v$ as
examined. The process repeats from Step $1$ to discover a new $v_{min}$. This
continues until the minimum is achieved for a node $v\in F$ or a node $v\in S$
and $b[v]\neq b[u]$.

\noindent\textbf{(Step $3$)} Each node $u$ sends a message $\langle
v,w,minval[u]\rangle$ to $s$ in the following manner. Each node $u$ waits until
it has received messages from all its children $v\in S,T$. Then, $u$ finds the
minimimum $minval[v], v\in N(u)$ where $N(u)=\{v:v\in S,T\text{ and }(u,v)\in
\E\}$. $u$ sends the message $\langle w,x,minval[u]\rangle$ to its parent where
$(w,x)$ is the edge for which the minimum was attained. For nodes $u$ in
blossoms, the parent depends on which set it was added to first. When $s$
receives all messages, it sends a message to nodes $w$ and $x$ where edge
$(w,x)$ is the one for which the minimum was attained. $w$ and $x$ mark each
other as examined.

The second stage of the algorithm is the \textbf{BLOSSOM} phase which executes
the blossom discovery process for an edge $(u,v)$ corresponding to a new
outermost blossom $Q$. This is done by sending messages from $u$ and $v$ towards
$s$ till the base of the blossom $b$ is found. The way in which these messages
are sent follows the description given by Shieber and Moran\cite{Shieber1986}.
Once this is done, $b$ sends messages throughout the blossom and informs all
nodes $u\in Q$ of their membership by assigning $b[u]=b$. In addition, for all
new nodes in the blossom, the corresponding parents and the corresponding
alternate distance $d_{s}$ or $d_{t}$ are assigned. 

The \textbf{GROW} phase extends the tree by adding one node to $T$ and one to
$S$. Let $(u,v)$ be the edge used to add $v$ to $T$ and $MATE(v)$ to $S$. $v$ is
added to $T$ and $MATE(v)$ to $S$. The $F$ labels are removed and the following
values are assigned:
\begin{IEEEeqnarray}{rCl}
d_{T}[v] &=& d_{S}[u]+w(u,v) \nonumber \\
d_{S}[MATE(v)] &=& d_{T}[v] \nonumber \\ 
P_{S}[MATE(v)] &=& v \nonumber \\
P_{T}[v] &=& u \nonumber
\end{IEEEeqnarray}

$MATE(v)$ sends a message $\langle v,d_{S}[MATE(v)]\rangle$ to all its neighbors
$w\in S$ informing them of its new $S$ label and its distance $d_{s}$.

The algorithm works in $O(n)$ iterations. In each one, it executes FINDMIN. Let
$\displaystyle \langle u,v,minval\rangle = \argmin_{u\in S} minval[u]$. If
$u,v\in S$, we execute BLOSSOM. Otherwise, we execute GROW. The overall steps of
the algorithm are outlined in Figure \ref{fig:algweighted}.

\begin{figure}[htbp]
\fbox{
\begin{minipage}{0.95\columnwidth}
\begin{algorithmic}

\STATE \textbf{Algorithm SHORTEST-CDC}
\vspace{2mm}
\STATE Set $s\in S$ and $\forall u\neq s$, $u\in F$

\STATE Set $d_{S}[s] = 0$ and $d_{S}[v] = d_{T}[v] = \infty$ for all $v\in F$

\STATE Execute FINDMIN. Let $(u,v)$ be the edge for which the minimum was
attained.

\IF{$u,v \in S$}

\STATE Execute BLOSSOM

\ELSE

\STATE Execute GROW

\ENDIF

\STATE Once $d$ is removed from F, the algorithm terminates.

\end{algorithmic}
\end{minipage}}
\caption{Algorithm SHORTEST-CDC for finding the shortest CDC path from $s$ to
$d$}

\label{fig:algweighted}
\end{figure}

We now prove the correctness of the algorithm by showing that, at each stage of
SHORTEST-CDC, the same edge is selected as in each phase of Edmonds' algorithm.
From now on, whenever we mention Edmonds' algorithm, we refer to each phase of
Edmonds' algorithm. 

\begin{fact}
\label{fac:matchzero}
For each edge $(u,v)\in \M$, $w(u,v)=0$
\end{fact}

We now give an outline of Edmonds' algorithm for the sake of completeness. For
more details, see \cite{Lawler1976}. Each node $u$ is associated with a dual
variable $y(u)$ and each odd subset of vertices $Q$ is assigned a dual variable
$z(Q)$. Note that $z(Q)>0$ only for blossoms and $z(Q)=0$ for every other odd
subset. Each node and blossom belong to one of three sets $\{S,T,F\}$.
Initially, $s\in S$ and everything else is in $F$. An edge is termed as tight if
$y(u)+y(v)+\sum_{u,v\in Q} z(Q)=0$. At each step, Edmonds' algorithm searches
for tight edges in order to close blossom or grow the tree. If there are no
tight edges, it makes a dual change. We choose $\delta=\min
(\delta_{1},\delta_{2},\delta_{3})$, where
\begin{IEEEeqnarray}{rCl}
\displaystyle
\delta_{1} & = & \min_{\text{non-trivial blossom }Q\in T} \frac{-z(Q)}{2}
\label{eqn:del3}\\
\delta_{2} & = & \min_{u\in S, v\in F} w(u,v)-\left(y(u)+y(v)\right) 
\label{eqn:del2}\\
\delta_{3} & = & \min_{u,v\in S}
\frac{w(u,v)-\left(y(u)+y(v)\right)}{2}\label{eqn:del1}
\end{IEEEeqnarray}
For each node $u\in S$, we set $y(u)=y(u)+\delta$ and for each node $u\in T$,
$y(u)=y(u)-\delta$. For each outermost blossom $Q\in S$ (resp. $Q\in T$),
$z(Q)=z(Q)-2\delta$(resp. $z(Q)=z(Q)+2\delta$). Now, the only case where a
blossom $Q\in T$ is when it shrunk at the start of the algorithm. In our case,
there are no shrunk blossoms at the beginning. Hence, we only need to worry
about $\delta_{2}$ and $\delta_{3}$.

\begin{lemma}
\label{lem:matchtight}
For each edge $(u,v)\in \M$, 
\[
\displaystyle{y(u)+y(v)+\sum_{u,v\in Q} z(Q)=0}
\]
throughout the course of Edmonds' algorithm.
\end{lemma}

\begin{proof}
Since each edge $(u,v)\in \M$ is tight throughout the course of the algorithm,
\begin{IEEEeqnarray}{rCl}
y(u)+y(v)+\sum_{u,v\in Q} z(Q) &=& w(u,v)
\end{IEEEeqnarray}
This, combined with Fact \ref{fac:matchzero} gives the desired result.
\end{proof}

Note that $d_{T}[u]$ is the weight of the shortest co-link path to $u$. This
leads to the following lemma whose proof is in the appendix.

\begin{lemma}
\label{lem:pathbounds}
For each node $u\in T$, 
\[
\displaystyle
d_{T}[u] = y(s)+y(u)+\sum_{u\in Q} z(Q)
\]
\end{lemma}

Using the above lemma, we can prove that our algorithm works in a similar manner
to Edmonds' algorithm.

\begin{lemma}
\label{lem:edmonds}
If edge $(u,v)$ becomes tight at some point in Edmonds' algorithm, it becomes
tight at the same point in the algorithm SHORTEST-CDC.
\end{lemma}

\begin{proof}
At some intermediate step of Edmonds' algorithm, let some value $\delta$ be
chosen as the dual change for that step. Let $y(u)$ be the dual variables of any
node $u$ before the dual change of $\delta$. Let $u'=MATE(u)$. Now, according to
Lemma \ref{lem:pathbounds},
\begin{IEEEeqnarray}{rCl}
\displaystyle
y(s) + y(u') + \sum_{u'\in Q} z(Q) & = & d_{T}[u'] \label{eqn:tightmate}
\end{IEEEeqnarray}
The value of $\delta$ is given by $\min(\delta_{1},\delta_{2})$ where
$\delta_{1}$ and $\delta_{2}$ are given by Equations \ref{eqn:del2} and
\ref{eqn:del1}. This is because, at the start of the algorithm, there are no
blossoms. Hence, no non-trivial $t$-blossoms are found and Equation
\ref{eqn:del3} never occurs. We may compute $\delta =
y(s)+\min(\delta_{1},\delta_{2}) = \min(y(s)+\delta_{1},y(s)+\delta_{2})$.

Taking the case of $\delta_{1}$, we have,
\begin{IEEEeqnarray}{rCl}
\displaystyle
\delta_{1} & = & \min_{u,v\in S} y(s) + \frac{w(u,v)-(y(u)+y(v))}{2} \\
& = & \min_{u,v\in S} \frac{w(u,v)+(y(s)-y(u))+(y(s)-y(v))}{2} \\
& = & \min_{u,v\in S} \frac{w(u,v)+d_{T}[u']+d_{T}[v']}{2} \label{eqn:delta1}
\end{IEEEeqnarray}
The above equations follow from Equation \ref{eqn:tightmate} and Lemma
\ref{lem:matchtight}. In the case of $\delta_{2}$, we may similarly compute it
as 
\begin{IEEEeqnarray}{rCl}
\delta_{2} & = & \min_{u\in S, v\in F} w(u,v)-(y(u)+y(v)) \\
& = & \min_{u\in S, v\in F} w(u,v)-y(u) \\
& = & \min_{u\in S, v\in F} w(u,v)+y(s)-y(u) \\
& = & \min_{u\in S, v\in F} d_{T}[u]+w(u,v) \label{eqn:delta2}
\end{IEEEeqnarray}
Since $d_{T}[u'] = d_{S}[u]$, Equations \ref{eqn:delta1} and \ref{eqn:delta2}
are exactly the values computed by the algorithm SHORTEST-CDC during the FINDMIN
phase. Hence, the lemma is proved.
\end{proof}

\begin{thm}
\label{thm:main}
The algorithm SHORTEST-CDC finds the shortest CDC path from $s$ to $t$ using
$O(n^{2})$ fixed-size messages.
\end{thm}

\begin{proof}
From Lemma \ref{lem:edmonds}, it is clear that the steps of Algorithm
SHORTEST-CDC follow the steps of Edmonds' algorithm exactly. Hence, we indeed
find the shortest augmenting path from $s$ to $t$. This implies that we have
found the shortest CDC path from $s$ to $t$.

We now analyze the communication complexity of Algorithm SHORTEST-CDC. There are
three main phases to the algorithm: (i) The FINDMIN phase. Here each node scans
its adjacent edges at most once through the algorithm. Hence, each edge is
scanned at most twice leading to a message complexity of $O(n^{2})$. (ii) The
GROW phase. Clearly, at each GROW phase, $O(1)$ messages are sent leading to a
total message complexity of $O(n)$. (iii) The BLOSSOM phase. Using the approach
of \cite{Shieber1986}, for each edge, we can achieve the backward and forward
processes using $O(n)$ messages. There can be no more than $O(n)$ blossoms found
during the course of the algorithm leading to a complexity of $O(n^{2})$.
Note that, although the approach of \cite{Shieber1986} is for unweighted
graphs, since we do not have blossom expansions, their method can be used
exactly since the rest of the blossom discovery is identical to unweighted
graphs \cite{Lawler1976}. Each message sent is of fixed size. Hence, the total
number of fixed-size messages sent during the algorithm is $O(n^{2})$.
\end{proof}

It is worth noting that if the graph is unweighted or if the edge weights are
integers, then we have several advantages. The complexity of FINDMIN is reduced
since there is no need to send messages back to $s$. The blossom discovery
process becomes simpler by discovering all blossom edges at the same time. We
can find a path using $O(n\frac{\pi}{R}+n\log{n})$ messages where $\pi$ is the
length of the shortest $(s-t)$ CDC path.

\section{Spanner for CDC Routing\label{spanner}}

In this section, we discuss a spanner for Channel-Discontinuity-Constraint
routing in wireless networks. We are given a graph $\G ( \V, \L)$ where $\V$ is
the set of nodes in $\Re^{2}$ and $\L$ is the set of links between nodes, and
the set of channels $C(u)$ available at each node $u$. Let $|u-v|$ denote the
euclidean distance between $u$ and $v$. The weight of every link $(u,v)$ is
given to be $|u-v|$. The cost of a path $\P$ is the sum of the weights of the
links along $P$ and is denoted by $|P|$.

\begin{defn}
Let $d_{\G}(u,v)$ denote the cost of the shortest CDC-path from $u$ to $v$ in
$\G$. If there is no such path, then the cost is infinity. We say that a graph
$\G'(\V, \L')$ where $\L'\subseteq \L$ is a \textbf{CDC $t$-Spanner} of $\V$
\emph{if and only if} for every $u,v \in \V$, $d_{\G'}(u,v)\leq t\cdot
d_{\G}(u,v)$ where $t$ is a constant.
\end{defn}

\begin{defn}
Let $\T=\{T_{1},...,T_{p}\}$ denote a partition of nodes in $\G$ into maximally
disjoint sets such that $\C(u)=\C(v)$ for every $u,v\in T_{i}$. We define the
\textbf{type} of a node $u$, denoted by $T(u)$, to be the set $T_{i}$ containing
$u$. If $C$ is the number of channels available in the network, then, $p\leq
2^{C}$.
\end{defn}

The above definition of CDC spanners is based on the definition of spanners for
Unit Disk Graphs in  \cite{wang-geometricspanner} but is equally relevant in the
case of networks using directional antennas. The CDC t-Spanner
$CDCYG_{k}(\V,\L',\C)$, which is based on the Yao graph~\cite{Yao}, is created
as follows. The set of links in $CDCYG_{k}$ is obtained in the following manner.
For each node $v\in \V$, divide the region around $v$ into $k$ interior-disjoint
sectors centered at $v$ with opening angle $\theta$. In each sector, connect $v$
to its two nearest neighbors of each type $T\in \T$ and connect these neighbors
with a link. This process can be performed in a distributed manner as described later. 
We prove that this graph is a CDC $t$-Spanner for the graph $\G$ where
$t=(1-2\sin{\tfrac{\theta}{2}})^{-2}$. This is particularly applicable when
directional antennas are used. Also, the computation of the above spanner does
not require that we know in advance the specific source $s$ and destination $d$
between which a CDC path needs to be computed.

\piccaption{$s$ and $d$ are located in clusters $A$ and $B$ respectively. The
only path from $s$ to $d$ is the edge $(s,d)$ which is not present in
$CDCYG_{k}$.\label{fig:oneedge}}
\parpic(2in,1in)[r][l]{
\includegraphics[width=2in]{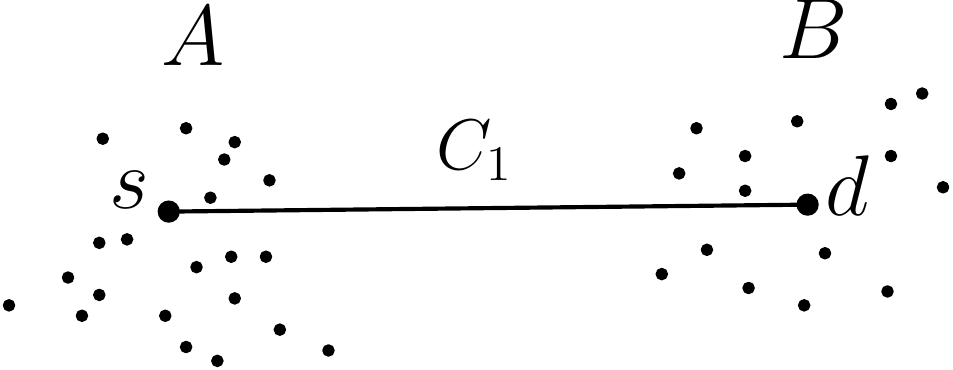}
} 
\begin{rem}
\label{rem:oneedge}

We note that no sub-quadratic size spanner can accommodate the degenerate case
where the the shortest CDC path between two nodes $s$ and $d$ is only one edge.
For example, in the case where $s$ and $d$ are in two clusters $A$ and $B$ and
there is only one channel shared between nodes in $A$ and nodes in $B$ (see
Figure \ref{fig:oneedge}), all edges from $A$ to $B$ may be required. However,
no routing algorithm is needed here since $s$ can check to see if $d$ is within
range and transmit to it.  When not considering the degenerate case, we obtain
the following theorem.
\end{rem}

\begin{thm}
\label{thm:sdpath}
$CDCYG_{k}$ is a CDC $t$-Spanner of $\G$ and $|\L'|=O(|V|)$ where
$t=(1-2\sin{\tfrac{\theta}{2}})^{-2}$.
\end{thm}

The following lemma was proven by Ruppert and Seidel~\cite{ruppert-yao} as a
part of their proof of the stretch factor of the Yao Graph.

\begin{lemma}[From \cite{ruppert-yao}]
\label{lem:boundpqr}
If there are two nodes $q$ and $r$ in a sector whose apex is $p$ and $|p-q|\leq
|p-r|$, then $|q-r|\leq |p-r|-(1-2\sin(\tfrac{\theta}{2}))|p-q|$.
\end{lemma}

\piccaption{Nodes $q$ and $r$ are in a sector of $p$ and
\\$|p-q|<|p-r|$.\label{fig:yaolem}}
\parpic(1.25in,0.9in)[r][l]{
\includegraphics[width=0.9in]{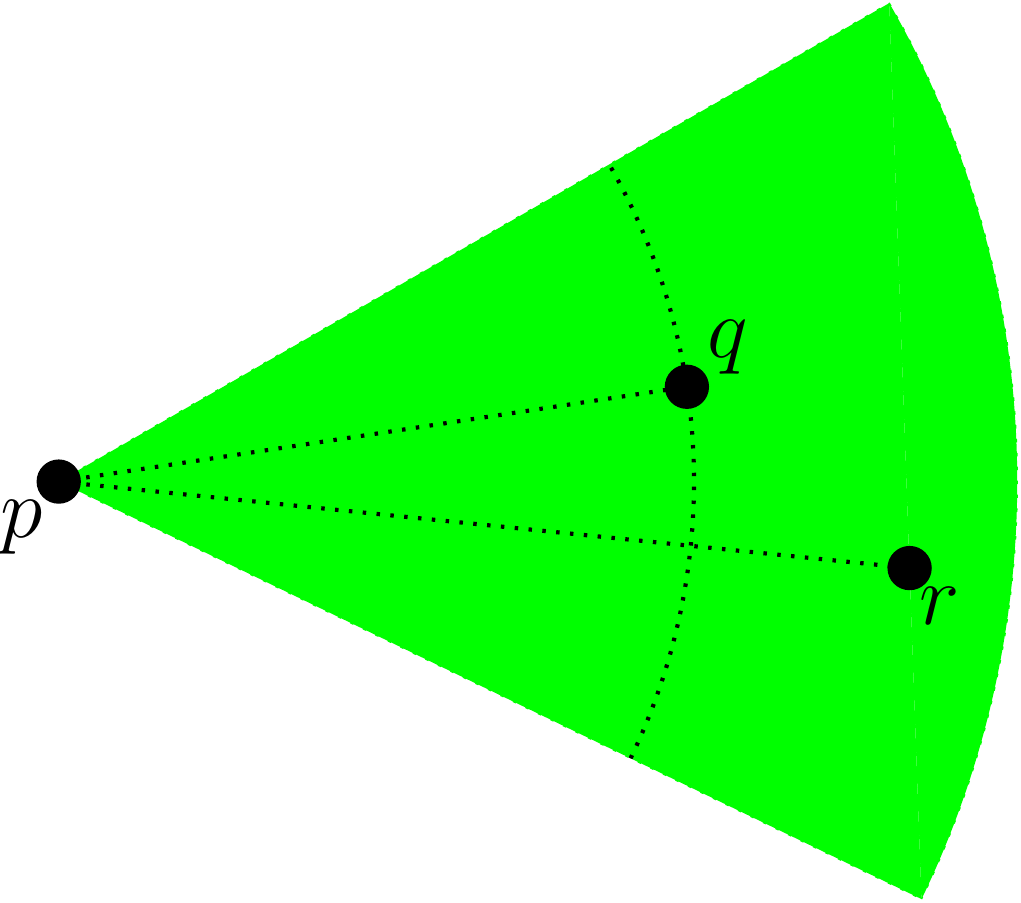}
}

\def\Pii{\P'_{{v_i} v_{i+1}}}

Let $\P:\{u_{0},u_{1},...,u_{m-1},u_{m}\}$ be an $(s,d)$ CDC-path in $\G$ where
$s=u_{0}$ and $d=u_{m}$. In the following lemmas, we assume that $m\geq 2$. The
case where $m=1$ is discussed in Remark~\ref{rem:oneedge} above. We describe a
procedure that shows that existence of an $(s,d)$ CDC-path $\P''$ in $CDCYG_{k}$
such that $|\P''|\leq t\cdot |\P|$. Note that this is not the algorithm used it
practice but is described just to show the existence of such a path. From now
on, for every path we create we also assign channels to its links.

The algorithm has two major components:

\noindent\textbf{\emph{Link Replacement:}} Each link $(u_{i}, u_{i+1})$ in $\P$
is replaced by a path. The resulting path $\P'$, obtained by replacing all the
links in $\P$, may not be simple. However, $|\P'|\leq t\cdot |\P|$.

\noindent\textbf{\emph{Untangling:}} A path $\P''$ is obtained from $\P'$ such
that $|\P''|\leq |\P'|$ and $\P''$ is simple.

We divide the links in $\P$ into three cases - (i) the intermediate links
$(u_{i},u_{i+1})$ for $i=1\dots m-2$, (ii) the first link $(s,u_{1})$ and (iii)
the last link $(u_{m-1},d)$.

Now, consider the first case where link $(u_{i},u_{i+1})\in \P$ is not the first
or last link. We replace this link with a path
$\P'_{u_{i},u_{i+1}}:\{u_{i}=v_{0},v_{1},...,v_{r-1},v_{r}=u_{i+1}\}$ in the
following manner. If, for even $j$, $(v_{j},u_{i+1})$ is a link in $\G$, the
next node is $u_{i+1}$. Otherwise, the next vertex $v_{j+1}$ of $\P$ is the
nearest neighbor of type either $T(u_{i})$ or $T(u_{i+1})$ in the sector $\Psi$
of $v_{j}$ where $\Psi$ is the sector containing $u_{i+1}$ . From $v_{j+1}$,
connect to the next nearest neighbor $v_{j+2}$ to $v_{j}$ of type $T(v_{j+1})$
the sector $\Psi$. Repeat for $v_{j+2}$. This process is demonstrated in Figure
\ref{fig:yaolinkspanner}. Let $C_{1},C_{2},C_{3}$ denote the channel assignment
in $\P$ to the links $(u_{i-1},u_{i})$, $(u_{i},u_{i+1})$ and
$(u_{i+1},u_{i+2})$ respectively. We may assign channels satisfying CDC to
$\P'_{u_{i},u_{i+1}}$ as follows (i) $C_{2}$ to $(v_{j},v_{j+1})$ and $C_{1}$ or
$C_{3}$ to $(v_{j+1},v_{j+2})$ for $j=0...r-3$ and (ii) $C_{2}$ to
$(v_{r-1},u_{i+1})$.

Now, for the second case, if $s$ has exactly one channel available to it, we may
only construct the above path through nodes of type $T(u_{1})$. In the third
case, if $d$ has exactly one channel available to it, then a path from $u_{m-1}$
to $d$ may not exist. However, we may construct a path from $d$ to $u_{m-1}$
through nodes of type $T(u_{m-1})$.

Let $\P'$ be the path obtained by concatenating $\P'_{u_{i+1},u_{i+2}}$ after 
$\P'_{u_{i},u_{i+1}}$ ($i=1\dots m-1$). We prove, in Lemma \ref{lem:pathexists},
that $\P'$ always exists. We also prove, in Lemma \ref{lem:linkspanner}, that
$|\P'|\leq t\cdot |\P|$ and we can assign channels to the links in $\P'$
satisfying CDC. However, $\P'$ might not be simple, since when constructing
$\P'_{u_{i},u_{i+1}}$, no caution is given for not using nodes of
$\P'_{u_{j},u_{j+1}}$ where $j\neq i$. This problem is resolved in Lemma
\ref{lem:untangling}, which shows how to create an $(s,d)$ CDC-path from $\P'$,
without increasing its length. The proofs of these lemmas are in the
appendix.

\begin{figure}[t]
\centerline{
\includegraphics[width=0.9\columnwidth]{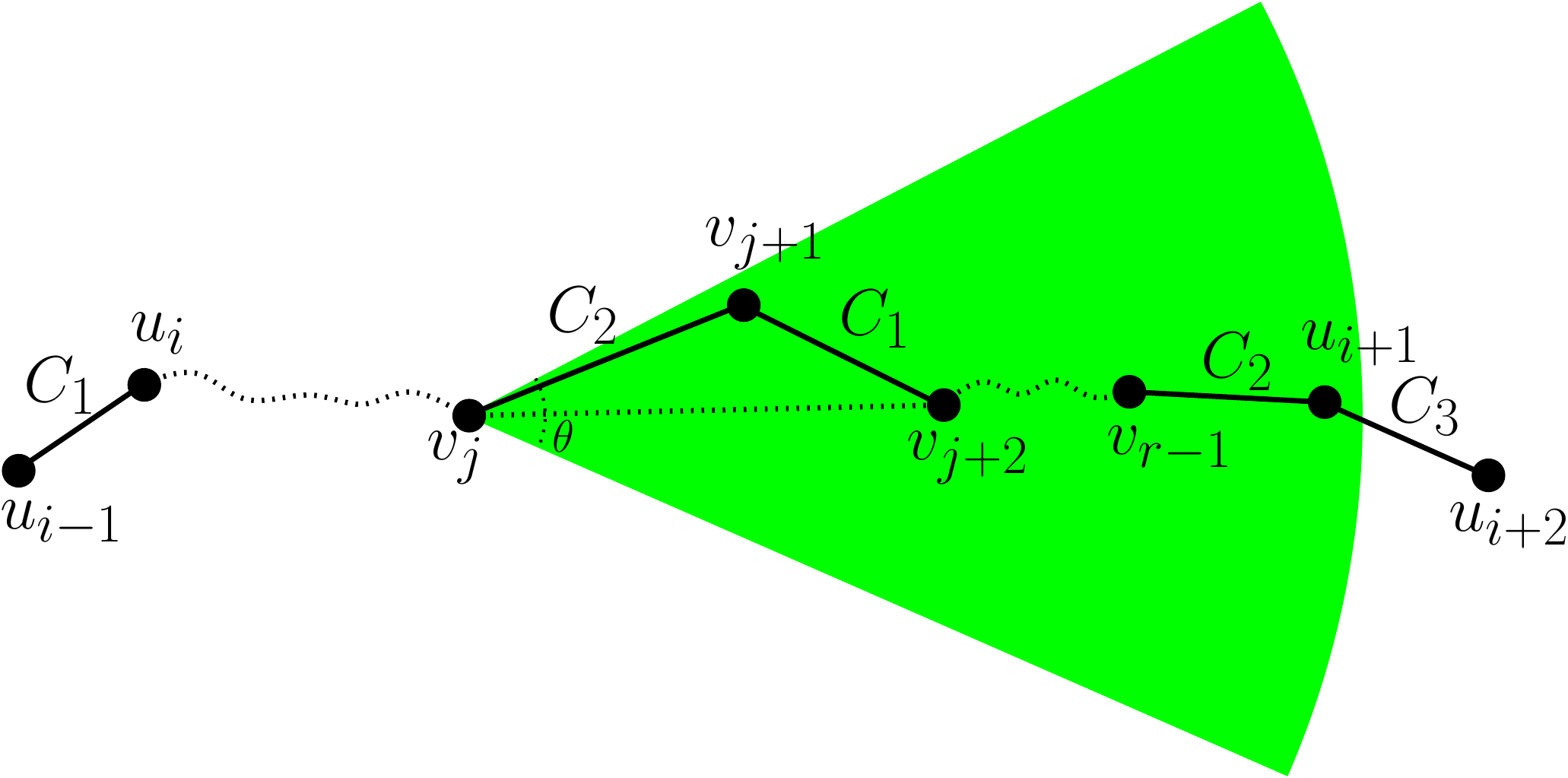}
}
\caption{A CDC-path from $u_{i}$ to $u_{i+1}$ constructed by repeatedly adding
the links $(v_{j},v_{j+1})$ and $(v_{j+1},v_{j+2})$ from $v_{j}$ until $u_{i+1}$
is a direct neighbor. $v_{j+1}$ and $v_{j+2}$ are the two nearest neighbors of
$v_{j}$ in the sector of $v_{j}$ containing $u_{i+1}$ whose set of channels is
the same as either $u_{i}$ or $u_{i+1}$.}
\label{fig:yaolinkspanner}
\end{figure}

\begin{lemma}
\label{lem:pathexists}
If there is a path $\P$ from $s$ to $d$ in $\G$, then the path $\P'$ from $s$ to
$d$ in $CDCYG_{k}$ constructed by the Link Replacement Procedure exists.
\end{lemma}

\begin{lemma}
\label{lem:linkspanner}
$|\P'|\leq t\cdot |\P|$ for $t=(1-2\sin{\tfrac{\theta}{2}})^{-2}$ and we can
assign channels to the links in $\P'$ satisfying the Channel Discontinuity
Constraint.
\end{lemma}

\begin{lemma}
\label{lem:untangling}
If, for some $i$, $\P'_{u_{i},u_{i+1}}$ overlaps with the portion of $\P'$
following $u_{i+1}$, then we can generate a new path $\P''$ such that (i)
$|\P''|\leq |\P'|$, (ii) has fewer such overlaps, and (iii) no two consecutive
links are assigned the same channel.
\end{lemma}

\begin{proof}[Proof of Theorem \ref{thm:sdpath}]
Consider a CDC-path $\P$ between nodes $s$ and $d$. The case where the number of
links in $\P$ is one has been discussed in Remark \ref{rem:oneedge}. If the
number of links in $\P$ is two or more, we may obtain the path $\P'$ in
$CDCYG_{k}$ as described in the link replacement procedure, which may not be
simple. From $\P'$, we may apply Lemma \ref{lem:untangling} to
$\P'_{u_{i},u_{i+1}}$ to remove any overlaps with $\P'_{u_{i+1},d}$ followed by
removal of overlaps with $\P'_{s,u_{i}}$ (for $i=1...m-1$), to obtain $\P''$, a
simple CDC-path in $CDCYG_{k}$ from $s$ to $d$. Note that $\P'_{s,u_{i}}$ and
$\P'_{u_{i+1},d}$ cannot overlap at the same node on $\P'_{u_{i},u_{i+1}}$ as
the original CDC-path is simple. Since we are eliminating all overlaps of
$\P'_{u_{j},u_{j+1}}$ for $j<i$ before $\P'_{u_{i},u_{i+1}}$, this case will not
occur because such a shared node will no longer be shared after the first time
it is found to overlap.  Hence, at each step, we reduce the number of overlaps
and at the final step when we are handling $\P'_{u_{m-1},u_{m}}$, there will be
no overlaps.

By Lemmas \ref{lem:linkspanner} and \ref{lem:untangling}, we know that
$|\P''|\leq |\P'|\leq t\cdot |\P|$. Hence, $CDCYG_{k}$ is a CDC $t$-Spanner for
for $t=(1-2\sin{\tfrac{\theta}{2}})^{-2}$. Now, for each node $v$ in
$CDCYG_{k}$, we have $k$ sectors. In each sector, there are links from $v$ to
two nodes of every type. Hence, there are $O(k|\T|)$ links for each node $v$. In
total, there are $O(k|\T||V|)$ links.
\end{proof}

Our results are of a theoretical nature, since the constants in the bounds might
be too large. However, we use them as evidence that, in practice, connecting
every node to a small number of its nearest neighbors, in each of the sectors
(lobes), either user-defined or explicitly by the directional antennas, will
yield a network with almost same properties as the original network.  Moreover,
a close look of the proof of Theorem \ref{thm:sdpath} reveals that, excluding
some convoluted cases, the bounds also hold if the cost of transmitting a
message between nodes $u,v$  is proportional to $|u-v|^\beta$, for any
$\beta>1$.

The total number of links in $CDCYG_{k}$ seems to be $O(k\cdot 2^{C}\cdot |V|)$
since the total number of types is $O(2^{C})$. We now argue that the number of
links can actually be bounded by $O(k\cdot C^{2}\cdot |V|)$ by constructing the
graph in the following manner. At a node $v$, for each pair $\{C_{a},C_{b}\}$ of
channels, connect to the two nearest neighbors which have these channels. Let
this graph be a CDC t-spanner $CDCYG'_{k}$. The proof of the following lemma is
in the appendix.

\begin{lemma}
\label{lem:quadbounds}
$CDCYG'_{k}$ is a CDC $t$-Spanner where
$t=\dfrac{1}{(1-2\sin{\tfrac{\theta}{2}})^{2}}$ and the number of links in
$CDCYG'_{k}$ is $O(k\cdot C^{2}\cdot |V|)$.
\end{lemma}

We now outline the algorithm to compute the spanner in a distributed manner using $O(n\log{n})$ messages.
Let $p$ be a node, and let $\{\Psi_{1}\dots \Psi_{\frac{\pi}{k}}\}$ be the set of sectors around $p$. For each
type $T$, $p$ needs to know the first two nearest neighbors of type $T$ in each sector $\Psi_{i}$. In order
to find these neighbors, $p$ sends a message, specifying the sector and type. Each node $q$ that hears the message, 
estimates the distance to $p$, and picks a random backoff time $\delta_{q}$ during which it waits and listen to see
how many nodes $w$ it hears whose distance $|w-p|<|p-q|$ . If when the timer of $q$ expires, it did not hear at least 
2 such nodes $w$, then $q$ broadcasts its distance $|q-p|$. For details on the message complexity of this approach, refer
to \cite{Arango2009}.

\section{Simulations}

We have conducted experiments to analyze the performance of SHORTEST-CDC under
several settings. We consider network topologies ranging of 100 to 800 nodes
distributed randomly in a region of $100\times 100$ unit$^{2}$. We analyze two
cases: (i) The density of the network is constant, i.e., the transmission range
is inversely proportional to the number of nodes in the network and (ii) The
transmission range is constant. In the first case, the transmission range varies
from 11 units for 100 nodes to 4 units for 800 nodes and in the second case, the
transmission range of the nodes is set at $5$ units. The number of channels
available in the network is either 10 or 20. The parameters of the system under
consideration are (i) The number of nodes and (ii) The number of channels
available at a node. We evaluate the following performance metrics:
\textbf{\emph{Communication Complexity:}} We evaluate the dependence of the
communication complexity (number of messages sent) of the algorithm on number of
nodes and channels. The worst-case complexity is $O(n^{2})$ in the expanded
graph which is a complexity of $O(C^{2}n^{2})$ in the original network. We are
interested in seeing which parameter has more effect on the communication
complexity and how many messages are sent in practice. 

\noindent\textbf{\emph{Number of Blossoms:}} Since the major portion of the
complexity is due to the formation and detection of blossoms in the graph, we
evaluate the number of blossoms which are formed on average and correlate this
with the communication complexity.
\begin{figure}[p]
\centering
\begin{minipage}{\textwidth}
\centering
\includegraphics[height=0.3\textheight]{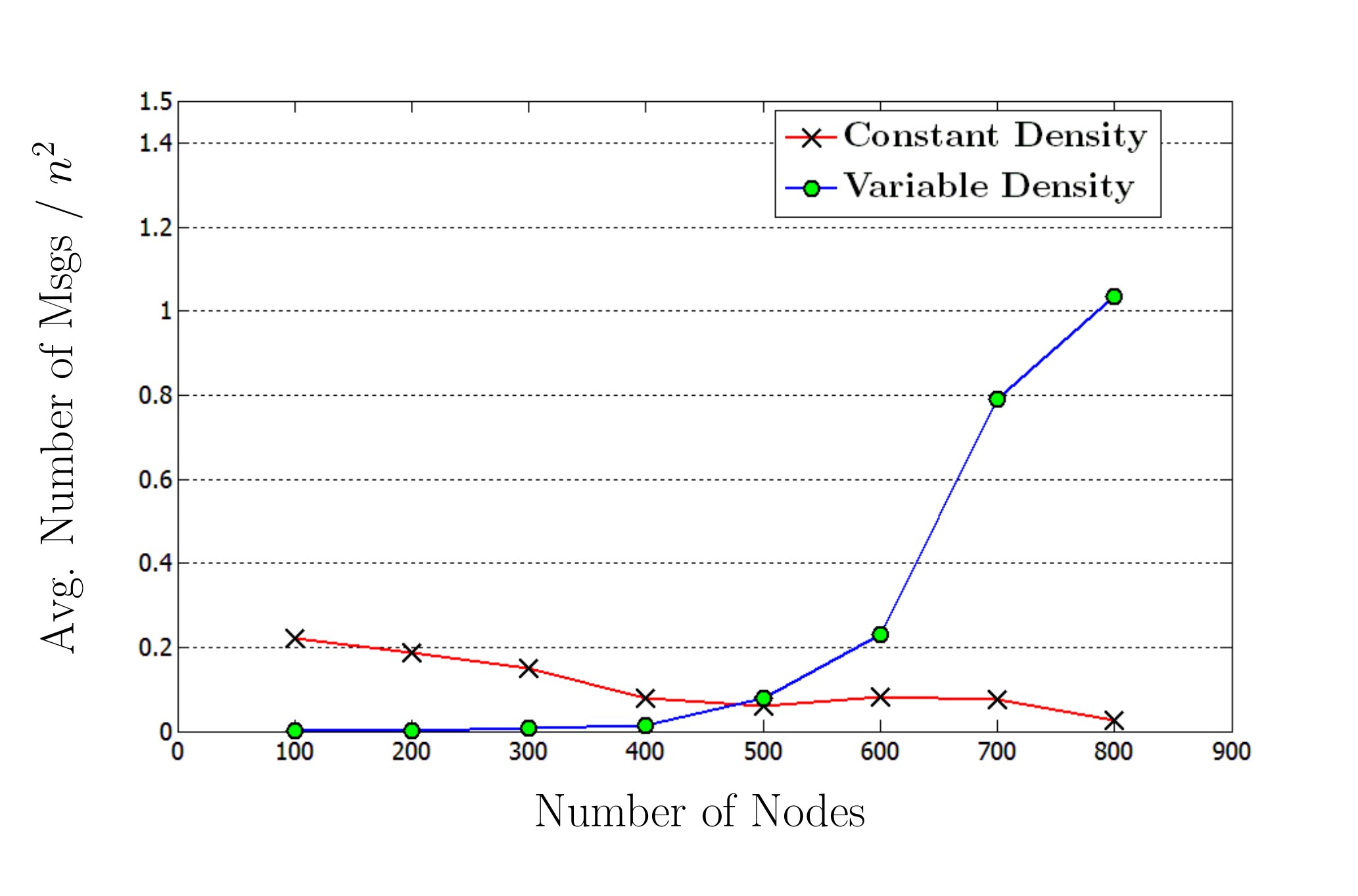}
\caption{Avg. Number of Messages$/n^{2}$ vs Number of nodes in the network}
\label{fig:resnvsmsg}
\end{minipage}
\begin{minipage}{\textwidth}
\centering
\includegraphics[height=0.3\textheight]{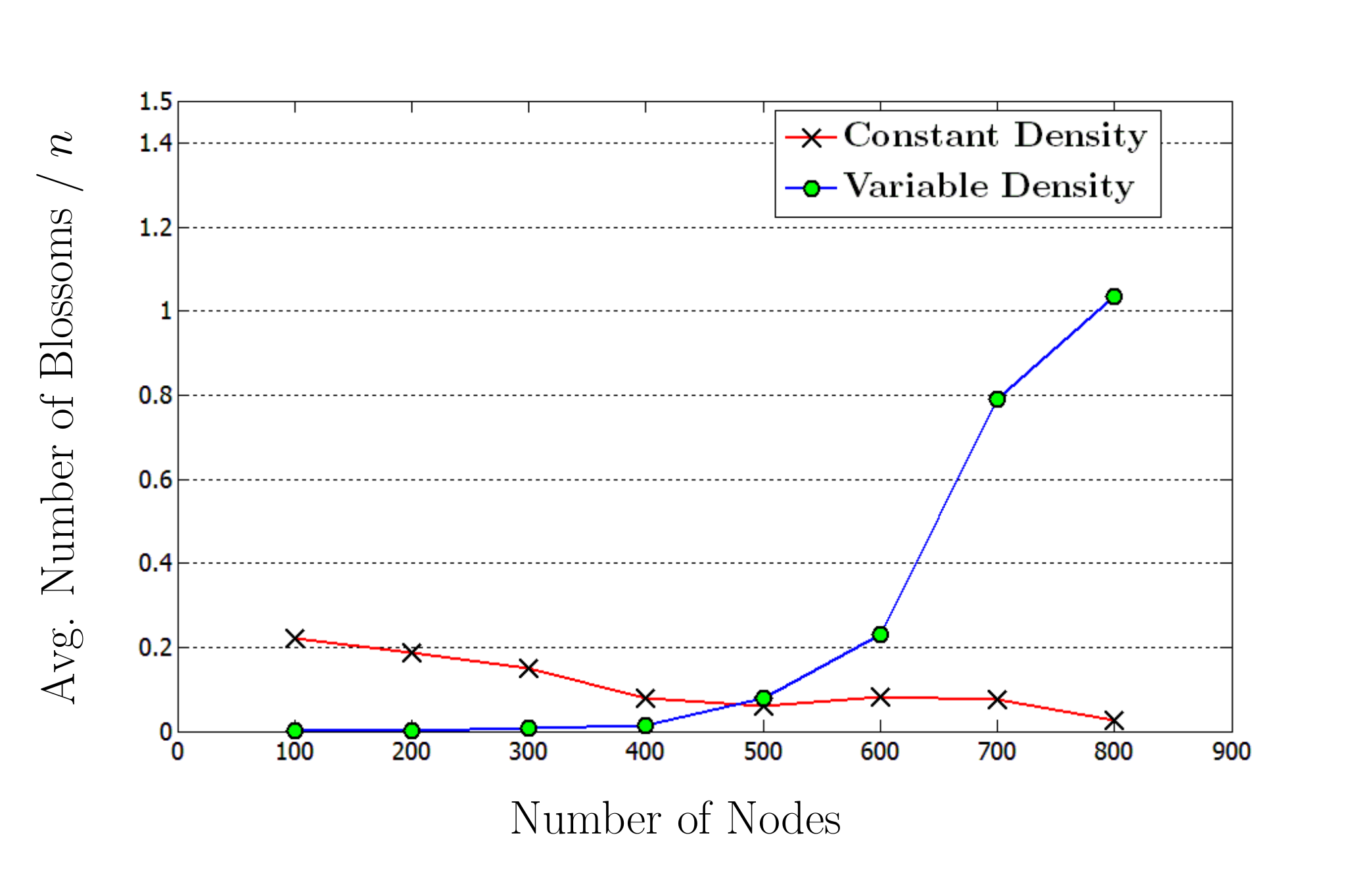}
\caption{Avg. number of blossoms$/n^{2}$ vs Number of Nodes}
\label{fig:resnvsbl}
\end{minipage}
\begin{minipage}{\textwidth}
\centering
\includegraphics[height=0.3\textheight]{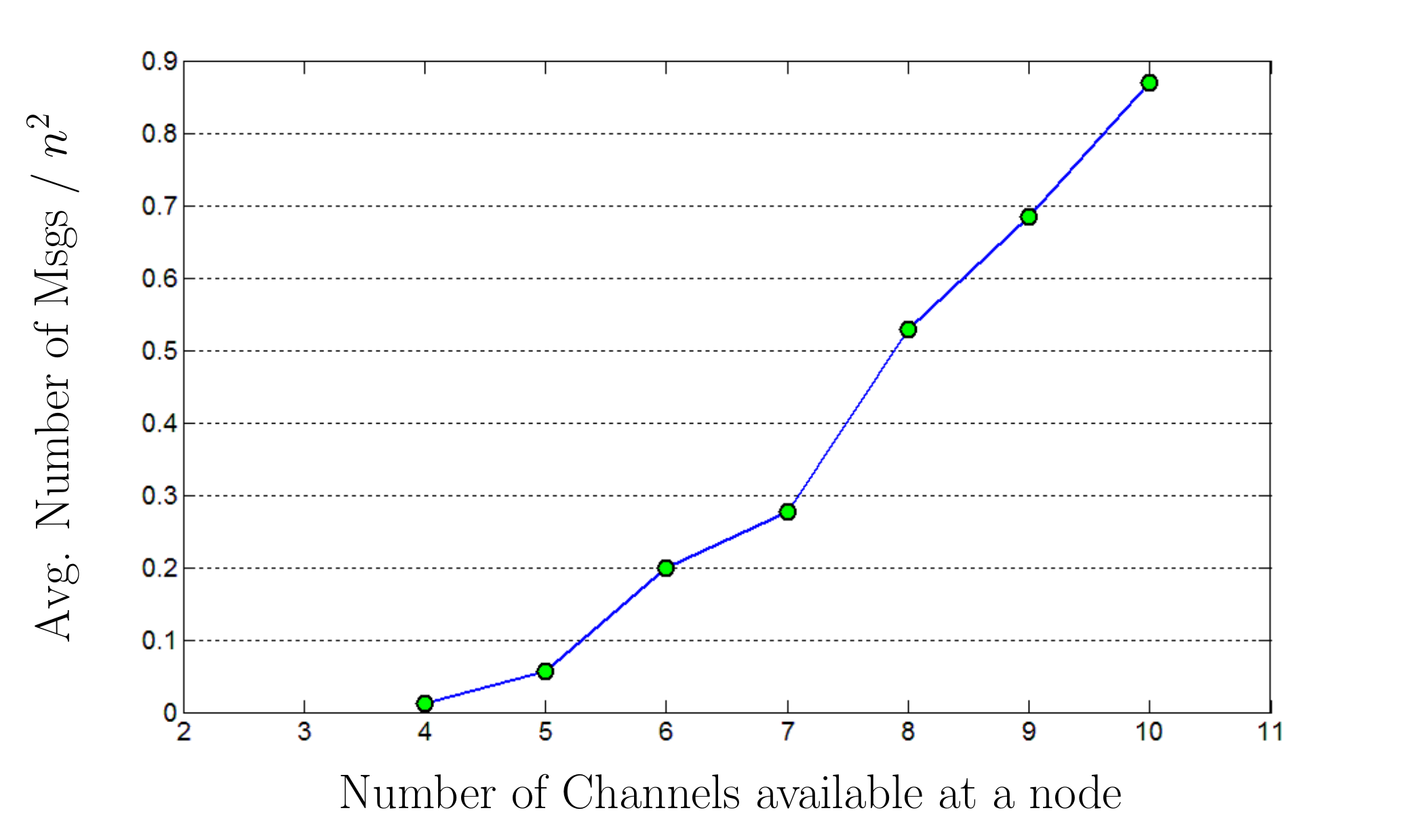}
\caption{Average number of Messages$/n^{2}$ vs Number of Channels available at a
node}
\label{fig:ressch}
\end{minipage}
\end{figure}

Figure \ref{fig:resnvsmsg} shows the ratio of the average number of messages
sent to $n^{2}$ as a function of the number of nodes in the network. Here, the
total number of channels is $10$ and each node picks $4$ channels. In case of
constant density, we can see that as the number of nodes increases, the ratio
decreases very gradually and hence, we may infer that the dependence on the number of nodes is limited. In the case of variable density (constant transmission range), we can see
that there is a steady increase in the ratio which seems to be quadratic in $n$.
This is expected since, the denser the network gets, the more the number of
blossoms and our intuition is that the blossoms are the major factor in the
complexity. In both cases, we can see that the number of messages does not stray
far beyond $n^{2}$ and is well below $C^{2}n^{2}$. Figure \ref{fig:resnvsbl}
shows the ratio of the average number of blossoms to the number of nodes $n$ as
a function of $n$ for the experiments run in Figure \ref{fig:resnvsmsg}. As can
be seen, it closely follows the graph of Figure \ref{fig:resnvsmsg}.  This leads
us to conclude that the number of blossoms is the major factor in the number of
messages.

Figure \ref{fig:ressch} shows the ratio of Average number of messages to $n^{2}$
as a function of the number of channels at a given node. In this case, $n$ is
fixed at $300$ nodes and the number of channels is varied from $4$ to $10$. The
function appears to be close to linear. This leads us to conjecture that, the
number of messages is, in practice, not quadratic in $C$ and could probably be
linear.

\section{Conclusion\label{sec:conc}}
In this paper, we studied the problem of path selection and channel
assignment with channel discontinuity constraint (CDC) in a wireless
infrastructure network.  Given that any two independent links may
share a channel, we showed that the problem of computing a CDC path is
equivalent of computing the minimum weight alternating path. We
developed a distributed algorithm to compute the minimum weight CDC
path with $O(n^{2})$ fixed-size messages.  Through experimental
evaluations on random network topologies, we study the trends on the
number of messages exchanged per path computation as the network size and number
of channels
increases.  In addition, we develop a $t$-spanner for CDC routing, where
the number of links employed is significantly reduced yet guaranteeing
that the minimum-weight CDC path in the spanner is no worse than $t$
times the minimum-weight CDC-path in the original network.

\section{Acknowledgements}

The work in this paper was partially supported by an NSF CAREER Grant 0348000.

\bibliographystyle{plain}
\bibliography{References}

\begin{thebibliography}{10}

\bibitem{BROADNETS08}
Sandeep~K. Ahuja, Abishek Gopalan, and Srinivasan Ramasubramanian.
\newblock Joint routing and channel assignment in multi-channel wireless
  infrastructure networks.
\newblock In {\em Proceedings of BROADNETS 2008}, 2008.

\bibitem{survey}
I.~F. Akyildiz, X.~Wang, and W.~Wang.
\newblock Wireless mesh networks: a survey.
\newblock In {\em Computer Networks}, 2004.

\bibitem{Bhatia}
M.~Alicherry, R.~Bhatia, and L.~E. Li.
\newblock Joint channel assignment and routing for throughput optimization in
  multiradio wireless mesh networks.
\newblock In {\em JSAC}, 2006.

\bibitem{Arango2009}
Jesus Arango, Alon Efrat, Srinivasan Ramasubramanian, Stephen Pink, and Marwan
  Krunz.
\newblock Retransmission and backoff strategies for wireless broadcasting.
\newblock {\em Ad Hoc Networks}, 8(1):77--95, 2010.

\bibitem{edge_3}
C.~L. Barrett, G.~Istrate, V.~S.~A. Kumar, M.~V. Marathe, S.~Thite, and
  S.~Thulasidasan.
\newblock Strong edge coloring for channel assignment in wireless radio
  networks.
\newblock In {\em Fourth Annual IEEE International Conference on Pervasive
  Computing and Communications Workshops}, 2006.

\bibitem{vertex_4}
A.~A. Bertossi, M.~C. Pinotti, and R.~Rizzi.
\newblock Channel assignment on strongly-simplicial graphs.
\newblock In {\em International Parallel and Distributed Processing Symposium},
  2003.

\bibitem{distance-2}
H.~L. Bodlaender, T.~Kloks, R.~B. Tan, and J.~V. Leeuwen.
\newblock $\lambda$-coloring of graphs.
\newblock In {\em The 17th Annual Symposium on Theoretical Aspects of Computer
  Science}, 2000.

\bibitem{Portal}
R.~Bruno, M.~Conti, and E.~Gregori.
\newblock Mesh networks: commodity multihop ad hoc networks.
\newblock In {\em IEEE Communications Magazine}, 2005.

\bibitem{DMesh}
S.~M. Das, H.~Pucha, D.~Koutsonikolas, Y.~C. Hu, and D.~Peroulis.
\newblock Dmesh: Incorporating practical directional antennas in multichannel
  wireless mesh networks.
\newblock In {\em JSAC}, 2006.

\bibitem{Draves}
R.~Draves, J.~Padhye, and B.~Zill.
\newblock Routing in multi-radio, multi-hop wireless mesh networks.
\newblock In {\em International Conference on Mobile Computing and Networking},
  2004.

\bibitem{Gabow}
H.~N. Gabow.
\newblock Data structures for weighted matching and nearest common ancestors
  with linking.
\newblock In {\em Proceedings of the 1st Annual ACM-SIAM Symposium on Discrete
  Algorithms}, 1990.

\bibitem{Gabow1974}
Harold~Neil Gabow.
\newblock {\em Implementation of algorithms for maximum matching on
  nonbipartite graphs.}
\newblock PhD thesis, Stanford University, 1974.

\bibitem{edge_2}
C.-C. Hsu, P.~Liu, D.-W. Wang, and J.-J. Wu.
\newblock Generalized edge coloring for channel assignment in wireless
  networks.
\newblock In {\em International Conference on Parallel Processing}, 2006.

\bibitem{BIBD}
H.-J. Huang, X.-L. Cao, X.-H. Jia, and X.-L. Wang.
\newblock A bibd-based channal assignment algorithm for multi-radio wireless
  mesh networks.
\newblock In {\em International Conference on Machine Learning and
  Cybernetics}, 2006.

\bibitem{mining}
G.~A. Kennedy and P.~J. Foster.
\newblock High resilience networks and microwave propagation in underground
  mines.
\newblock In {\em The 9th European Conference on Wireless Technology}, 2006.

\bibitem{Kodialam}
M.~Kodialam and T.~Nandagopal.
\newblock Characterizing achievable rates in multi-hop wireless networks: the
  joint routing and scheduling problem.
\newblock In {\em Proceedings of the ACM Mobicom}, 2003.

\bibitem{vertex_3}
S.~O. Krumke, M.~V. Marathe, and S.~S. Ravi.
\newblock Models and approximation algorithms for channel assignment in radio
  networks.
\newblock In {\em Wireless Networks}, 2001.

\bibitem{Lawler1976}
E~Lawler.
\newblock {\em Combinatorial Optimization: Networks and Matroids}.
\newblock Holt, Rinehart, Winston, 1976.

\bibitem{video_streaming}
N.~Mastronarde, Y.~Andreopoulos, M.~V.~D. Schaar, D.~Krishnaswamy, and
  J.~Vicente.
\newblock Cross-layer video streaming over 802.11e-enabled wireless mesh
  networks.
\newblock In {\em ICASSP}, 2006.

\bibitem{Meng}
X.~Meng, K.~Tan, and Q.~Zhang.
\newblock Joint routing and channel assignment in multi-radio wireless mesh
  networks.
\newblock In {\em ICC}, 2006.

\bibitem{vertex_2}
K.~N. Ramachandran, E.~M. Belding, K.~C. Almeroth, and M.~M. Buddhikot.
\newblock Interference-aware channel assignment in multi-radio wireless mesh
  networks.
\newblock In {\em INFOCOM}, 2006.

\bibitem{Ramanathan2}
R.~Ramanathan, J.~Redi, C.~Santivanez, D.~Wiggins, and S.~Polit.
\newblock Ad hoc networking with directional antennas: A complete system
  solution.
\newblock {\em JSAC}, 23:496--506, 2005.

\bibitem{Raniwala2}
A.~Raniwala and T.-C. Chiueh.
\newblock Architecture and algorithms for an ieee 802.11-based multi-channel
  wireless mesh network.
\newblock In {\em INFOCOM}, 2005.

\bibitem{Raniwala1}
A.~Raniwala, K.~Gopalan, and T.-C. Chiueh.
\newblock Centralized channel assignment and routing algorithms for
  multi-channel wireless mesh networks.
\newblock In {\em ACM SIGMOBILE Mobile Computing and Communications Review},
  2004.

\bibitem{ruppert-yao}
J.~Ruppert and R.~Seidel.
\newblock Approximating the $d$-dimensional complete {Euclidean} graph.
\newblock In {\em Proc. 3rd Canad. Conf. Comput. Geom.}, 1991.

\bibitem{Shieber1986}
Baruch Shieber and Shlomo Moran.
\newblock Slowing sequential algorithms for obtaining fast distributed and
  parallel algorithms: maximum matchings.
\newblock In {\em PODC '86: Proceedings of the fifth annual ACM symposium on
  Principles of distributed computing}, pages 282--292, New York, NY, USA,
  1986. ACM.

\bibitem{SAFE}
M.~Shin, S.~Lee, and Y.-A. Kim.
\newblock Distributed channel assignment for multi-radio wireless networks.
\newblock In {\em MASS}, 2006.

\bibitem{capacity_bound}
A.~Spyropoulos and C.~S. Raghavendra.
\newblock Capacity bounds for ad-hoc networks using directional antennas.
\newblock In {\em IEEE International Conference on Communications}, 2003.

\bibitem{iAWARE}
A.~P. Subramanian, M.~M. Buddhikot, and S.~Miller.
\newblock Interference aware routing in multi-radio wireless mesh networks.
\newblock In {\em 2nd IEEE Workshop on Wireless Mesh Networks}, 2006.

\bibitem{vertex_1}
A.~P. Subramanian, H.~Gupta, and S.~R. Das.
\newblock Minimum interference channel assignment in multi-radio wireless mesh
  networks.
\newblock In {\em SECON}, 2007.

\bibitem{kasturi-divconq}
Kasturi~R. Varadarajan.
\newblock A divide-and-conquer algorithm for min-cost perfect matching in the
  plane.
\newblock In {\em Proceedings of FOCS '98}, 1998.

\bibitem{kasturi-approx}
Kasturi~R. Varadarajan and Pankaj~K. Agarwal.
\newblock Approximation algorithms for bipartite and non-bipartite matching in
  the plane.
\newblock In {\em Proceedings of SODA '99}, 1999.

\bibitem{wang-geometricspanner}
Yu~Wang and Xiang-Yang Li.
\newblock Distributed spanner with bounded degree for wireless ad hoc networks.
\newblock In {\em IPDPS '02: Proceedings of the 16th International Parallel and
  Distributed Processing Symposium}, 2002.

\bibitem{edge_1}
L.~Xu, Y.~Xiang, and M.~Shi.
\newblock A novel channel assignment algorithm based on topology simplification
  in multi-radio wirelesss mesh networks.
\newblock In {\em 25th IEEE International Performance, Computing, and
  Communications Conference}, 2006.

\bibitem{Yao}
Andrew~C Yao.
\newblock On constructing minimum spanning trees in k-dimensional spaces and
  related problems.
\newblock Technical report, Stanford University, 1977.

\bibitem{capacity_improve}
S.~Yi, Y.~Pei, and S.~Kalyanaraman.
\newblock On the capacity improvement of ad hoc wireless networks using
  directional antennas.
\newblock In {\em ACM international symposium on Mobile ad hoc networking and
  computing}, 2003.

\bibitem{capacity_improvement}
J.~Zhang and S.~C. Liew.
\newblock Capacity improvement of wireless ad hoc networks with directional
  antennae.
\newblock In {\em IEEE Vehicular Technology Conference}, 2006.

\bibitem{pnc}
Shengli Zhang, Soung~Chang Liew, and Patrick~P. Lam.
\newblock Hot topic: physical-layer network coding.
\newblock In {\em MobiCom '06: Proceedings of the 12th annual international
  conference on Mobile computing and networking}, 2006.

\end{thebibliography}

\appendix

\section{Proof of Lemma \ref{lem:pathbounds}}

\medskip
{\bf Lemma \ref{lem:pathbounds}.}\ \ 
{\sl For each node $u\in T$, 
\[
\displaystyle
d_{T}[u] = y(s)+y(u)+\sum_{u\in Q} z(Q)
\]
}

\begin{proof}
The proof of this lemma follows from the description given by Gabow
\cite{Gabow1974}. Let the shortest alternating path from $s$ tox $u$ ending in an
unmatched edge be $\pi = \{u_{0}=s,u_{1},u_{2},\dots,u_{k}=u\}$. The weight of
this path $W(s,u)$ is given by,

\begin{IEEEeqnarray}{rCl}
\displaystyle
W(s,u) & = & \sum_{i=0}^{k-1} w(u_{i},u_{i+1}) \\
& = & \left(w(s,u_{1})+w(u_{2},u_{3})+\dots+w(u_{k-1},u)\right) + \nonumber \\ 
& & \left(w(u_{1},u_{2})+w(u_{3},u_{4})+\dots+w(u_{k-2},u_{k-1})\right) \nonumber \\
& & \\
& = & \left(w(s,u_{1})+w(u_{2},u_{3})+\dots+w(u_{k-1},u)\right) - \nonumber \\ 
& & \left(w(u_{1},u_{2})+w(u_{3},u_{4})+\dots+w(u_{k-2},u_{k-1})\right) \nonumber \\
& & \label{eqn:pathlength}
\end{IEEEeqnarray}
by Fact \ref{fac:matchzero}, since the weights of the matched edges is zero.

Since each edge $(u_{i},u_{i+1})$ is tight, we have
$\displaystyle{y(u_{i})+y(u_{i+1})+\sum_{u_{i},u_{i+1}\in Q} z(Q) =
w(u_{i},u_{i+1})}$. Substituting this into Equation \ref{eqn:pathlength}, we get
the result of the lemma, i.e.,

\begin{IEEEeqnarray}{rCl}
\displaystyle
d_{T}[u] & = & y(s)+y(u)+\sum_{u\in Q} z(Q)
\end{IEEEeqnarray}
since $s$ is not part of any blossoms(it can only be the base of some blossom).
\end{proof}

\section{Proof of Lemmas from Section \ref{spanner}}

{\bf Lemma \ref{lem:pathexists}.}\ \ 
{\sl If there is a path $\P$ from $s$ to $d$ in $\G$, then the path $\P'$ from
$s$ to $d$ in $CDCYG_{k}$ constructed by the Link Replacement Procedure exists.
}
\begin{proof}
Consider a link $(u,v)\in \P$. We first look at the case where $(u,v)$ is an
intermediate link. Consider a step of the Link Replacement procedure where we
are at a node $v_{j}$, for some even $j$. Now, in the sector of $v_{j}$
containing $v$, we check if the $v$ is among the two nearest neighbors of type
$T(v)$. If this is the case, we may directly connect to it. Now, since $v$ is a
neighbor of $u$ and for each $v_{j}$, $|v_{j}-v|\leq |v_{j-2}-v|$, we have the
fact that $v$ is a neighbor of $v_{j}$. Hence, there cannot be a case that
$v_{j}$ does not have any neighbors of type $T(v)$. So, we will always be able
to reach $v$ and hence, there is a path from $u$ to $v$ in $CDCYG_{k}$.

In the case where $(u,v)$ is the first link, we may be restricted to paths
through nodes of type $T(v)$ and the above reasoning holds good in this case
also. In the case where $(u,v)$ is the last link and when $v$ has only one
channel, say $C_{1}$, we may not find a path from $u$ to $v$. If we try to find
a path from $u$ to $v$ through nodes of type $T(u)$, then we may get stuck if
$u$ has no neighbors of type $T(u)$. On the other hand, if we find a path
through nodes of type $T(v)$, we can only assign channel $C_{1}$ to links along
this path and hence, it will not satisfy CDC. However, we know that $T(u)$ has
more than one channel since we are assuming that $m\geq 3$ and hence, there are
two or more edges in $\P$. Therefore, we may find a path from $v$ to $u$ through
nodes of type $T(u)$.

\end{proof}

{\bf Lemma \ref{lem:linkspanner}.}\ \ 
{\sl $|\P'|\leq t\cdot |\P|$ for $t=(1-2\sin{\tfrac{\theta}{2}})^{-2}$ and we
can assign channels to the links in $\P'$ satisfying the Channel Discontinuity
Constraint.
}
\begin{proof}
Let $\P'_{u_{i},u_{i+1}}:\{u_{i}=v_{0},v_{1},...,v_{r-1},v_{r}=u_{i+1}\}$ be the
portion of $\P'$ from $u_{i}$ to $u_{i+1}$ for some $i$. Consider the path from
$u_{i}$ to $u_{i+1}$ $\pi:\{u=v_{0},v_{2},...,v_{r-3},v_{r-1},v_{r}=u_{i+1}\}$
consisting of nodes $v_{j}$, where $j$ is even, in $\P'_{u_{i},u_{i+1}}$
together with the node $u_{i+1}$. Note that $\pi$ is not necessarily a CDC-path.
The links $(v_{j},v_{j+2})$ along with the link $(v_{r-1},u_{i+1})$ as shown in
Figure \ref{fig:yaolinkspanner} constitute the links in $\pi$. For each node
$v_{j}$ on this path, from Lemma \ref{lem:boundpqr}, we know that
\begin{IEEEeqnarray}{l}
|v_{j+2}-u_{i+1}| \leq |v_{j}-u_{i+1}| - (1-2\sin{\tfrac{\theta}{2}})|v_{j}-v_{j+2}| \label{eqn:boundvi}
\end{IEEEeqnarray}
See Figure \ref{fig:yaolinkspanner}. Now, by summation over all even $j$, we get
\begin{IEEEeqnarray}{l}
\displaystyle
\sum_{\substack{j=0\\ j\mbox{ \small{is even}}}}^{r-3} |v_{j+2}-u_{i+1}| \nonumber \\
\leq \sum_{\substack{j=0\\ j\mbox{ \small{is even}}}}^{r-3} |v_{j}-u_{i+1}|-\sum_{\substack{j=0\\ j\mbox{ \small{is even}}}}^{r-3}(1-2\sin{\tfrac{\theta}{2}})|v_{j}-v_{j+2}| \nonumber \\
~
\end{IEEEeqnarray}
Rearranging (details omitted), we get
\begin{IEEEeqnarray}{l}
\displaystyle
\sum_{\substack{j=0\\ j\mbox{ \small{is even}}}}^{r-3} |v_{j}-v_{j+2}| \leq c\cdot |u_{i}-u_{i+1}|\label{eqn:evenpathbound}
\end{IEEEeqnarray}
where $c=(1-2\sin{\tfrac{\theta}{2}})^{-2}$.
Now, we have in $\pi$, for each $v_{j},v_{j+2}$, where $j$ is even and $0\leq
j\leq r-3$, from Lemma \ref{lem:boundpqr},
\begin{IEEEeqnarray}{rCl}
\displaystyle
|v_{j+1}-v_{j+2}| &\leq & |v_{j}-v_{j+2}| - \dfrac{1}{c}\cdot |v_{j}-v_{j+1}| \nonumber \\
\Longrightarrow c\cdot |v_{j}-v_{j+2}| & \geq & |v_{j}-v_{j+1}| + |v_{j+1}-v_{j+2}|  \label{eqn:eachedgebound}
\end{IEEEeqnarray}
since $\dfrac{1}{c}\leq 1$. Combining equations \ref{eqn:evenpathbound} and
\ref{eqn:eachedgebound}, we get . Hence,
\begin{align}
\displaystyle
\sum_{j=0}^{r-2} |v_{j}-v_{j+1}| &\leq c\cdot \sum_{\substack{j=0\\ j\mbox{
\small{is even}}}}^{r-3} |v_{j}-v_{j+2}| \nonumber \\
\sum_{j=0}^{r-1} |v_{j}-v_{j+1}| &\leq c\cdot 
|v_{r-1}-u_{i+1}|+\sum_{\substack{j=0\\ j\mbox{ \small{is even}}}}^{r-3} 
|v_{j}-v_{j+2}| \nonumber \\
&\leq c^{2}\cdot |u_{i}-u_{i+1}|
\end{align}
Hence, $|\P'_{u_{i},u_{i+1}}|\leq t\cdot |u_{i}-u_{i+1}|$ for
$t=(1-2\sin{\tfrac{\theta}{2}})^{-2}$. Consider that the nodes $u_{i}$ and
$u_{i+1}$ share the channel $C_{2}$ and that the links $(u_{i-1},u_{i})$ and
$(u_{i},u_{i+1})$ are assigned the channels $C_{1}$ and $C_{3}$ respectively. We
now have a channel assignment for the path $\pi$ because for every set of nodes
$v_{j},v_{j+1},v_{j+2}$ for $i=0...r-3$ along this path when $i$ is even, we can
assign channels $C_{2},C_{1}$ or $C_{2},C_{3}$ depending on whether $v_{j+1}$
and $v_{j+2}$ are of type $T(u_{i})$ or $T(u_{i+1})$. Finally, we assign the
channel $C_{2}$ to the link $v_{r-1},u_{i+1}$. This channel assignment is shown
in Figure ~\ref{fig:yaolinkspanner}.
Hence, repeating this over all links in $\P'$, we have a channel assignment for
links in $\P'$ and $|\P'| \leq t\cdot |\P|$, for
$t=(1-2\sin{\tfrac{\theta}{2}})^{-2}$.
\end{proof}

\begin{figure*}[t]
\centering

\subfloat[\sl{The removal of overlaps when $C_{out}\neq C_{2}$. The path
\newline $\P'':\{s,...,u_{i},...,v_{j-1},v_{j},w,...,d\}$ (marked as a BLUE
dashed path) removes the portion of $\P'$ from $v_{j}$ back to itself.}]{
\includegraphics[width=0.4\textwidth]{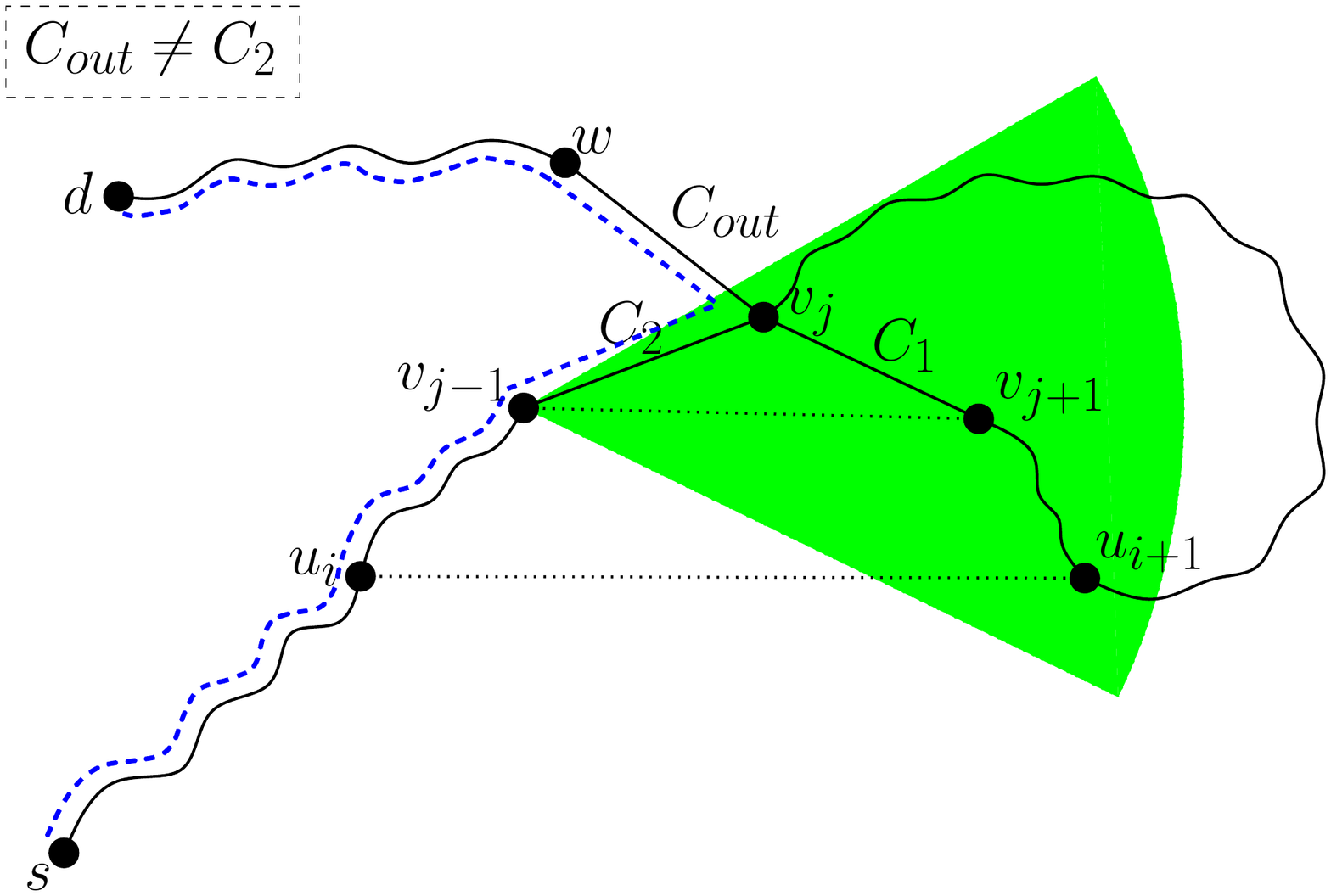}
\label{fig:overlap1}
}
\hfil
\subfloat[\sl{The removal of overlaps when $C_{out}=C_{2}$. The path \newline
$\P'':\{s,...,u_{i},...,v_{j-1},v_{j+1},v_{j},w,...,d\}$ (marked as a BLUE
dashed path) consists of the portion of $\P'$ from $s$ to $v_{j-1}$, the links
$(v_{j-1},v_{j+1})$ and $(v_{j+1},v_{j})$ followed by the portion of $\P'$ from
$v_{j}$ to $d$.}]{
\includegraphics[width=0.4\textwidth]{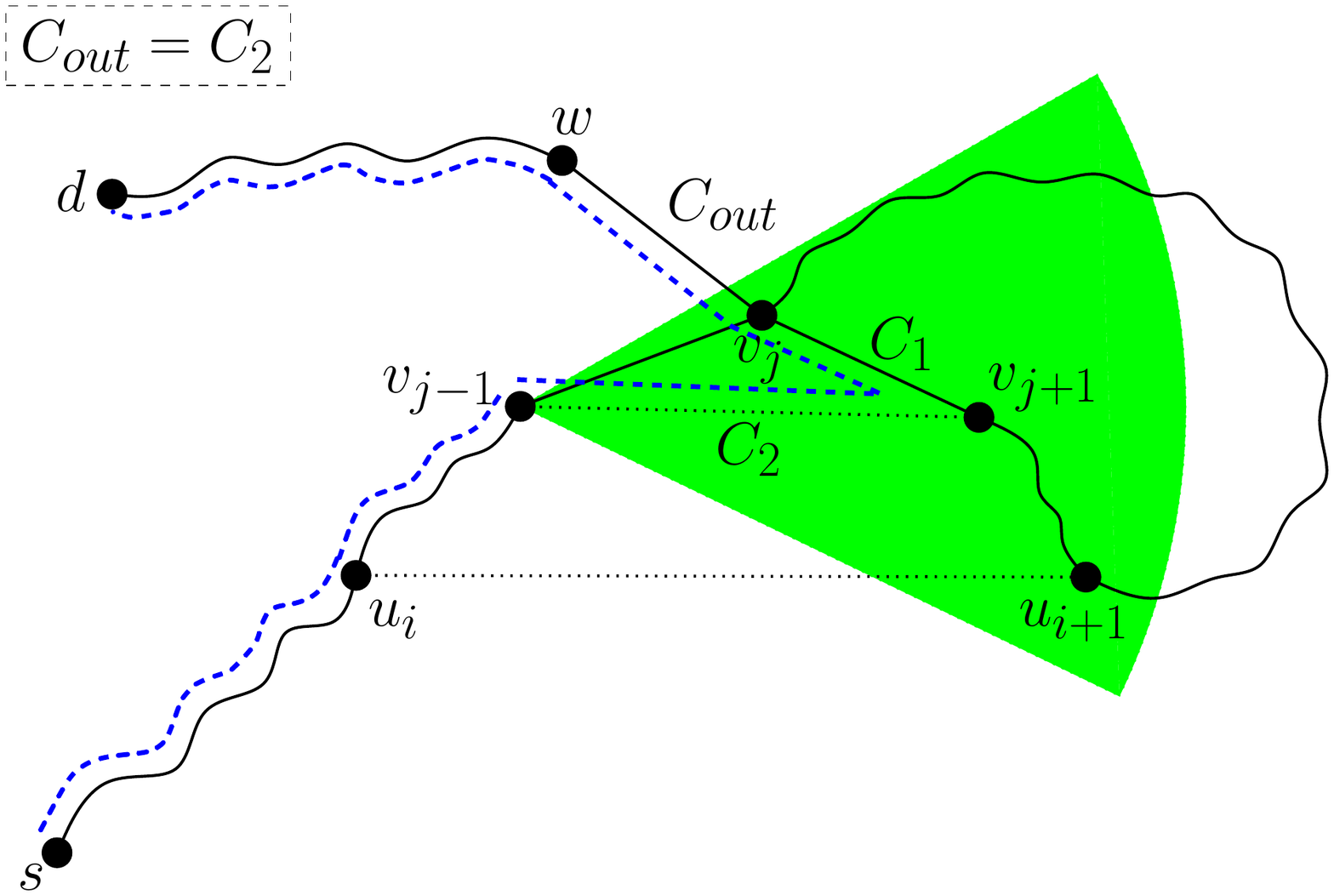}
\label{fig:overlap2}
}

\caption{The case where $\P'$ intersects itself at node $v_{j}$, in
$\P'_{u_{i},u_{i+1}}$, where $j$ is odd. Path $\P'$ is shown as a BLACK solid
path and $\P''$ which does not intersect itself at any node in
$\P'_{u_{i},u_{i+1}}$ is shown as a BLUE dashed path. }
\label{fig:overlap}
\end{figure*}

{\bf Lemma \ref{lem:untangling}.}\ \ 
{\sl If, for some $i$, $\P'_{u_{i},u_{i+1}}$ overlaps with the portion of $\P'$
following $u_{i+1}$, then we can generate a new path $\P''$ such that (i)
$|\P''|\leq |\P'|$, (ii) has fewer such overlaps, and (iii) no two consecutive
links are assigned the same channel.
}
\begin{proof}
Let the portion of $\P'$ preceding $u_{i}$ be $\P'_{s,u_{i}}$ and that following
$u_{i+1}$ be $\P'_{u_{i+1},d}$. Now, either $\P'_{s,u_{i}}$ or $\P'_{u_{i+1},d}$
or both share nodes with $\P'_{u_{i},u_{i+1}}$. We can obtain a new CDC-path
$\P''$ as follows.

We show the CDC-path for the case where $\P'_{u_{i+1},d}$ shares nodes with
$\P'_{u_{i},u_{i+1}}$. Assume that $\P'_{u_{i+1},d}$ uses a node $v_{j}$ in
$\P'_{u_{i},u_{i+1}}$ and this is the last node in $\P'_{u_{i+1},d}$ which is in
$\P'_{u_{i},u_{i+1}}$. Let the channel assigned to the outgoing link from
$v_{j}$ to a node $w$ in the path $\P'$ be $C_{out}$ as shown in Figure
\ref{fig:overlap}. There are three cases for $C_{out}$ - (i) $C_{out}\neq C_{1}$
and $C_{out}\neq C_{2}$, (ii) $C_{out}=C_{1}$ and (iii) $C_{out}=C_{2}$. Now, we
obtain a new path (not necessarily simple) $\P''$ satisfying CDC by replacing
the portion of $\P'$ from $u_{i}$ to $v_{j}$ with a new path
$\P''_{u_{i},v_{j}}$ which we describe below.

There are two cases for the node $v_{j}$ depending on whether $j$ is odd or
even. We first handle the case when $j$ is odd. In this case, $v_{j}$ and
$v_{j+1}$ are in the cone of $v_{j-1}$ which contains $u_{i+1}$. This is
because, according to the link replacement procedure, from $v_{j-1}$, we would
have picked the two nearest neighbors in the cone of $v_{j-1}$ containing
$u_{i+1}$. If $j$ is odd, the two cases where $C_{out}\neq C_{2}$ and
$C_{out}=C_{2}$ are depicted in Figure \ref{fig:overlap}.

If $C_{out}\neq C_{2}$, we obtain $\P''_{u_{i},v_{j}}$ by going from $u_{i}$ to
$v_{j}$ along $\P'_{u_{i},u_{i+1}}$. The $(s,d)$ path obtained using this
$\P_{u_{i},v_{j}}$ satisfies CDC because, according to the channel assignments
for $\P'_{u_{i},u_{i+1}}$, the link $(u_{i},v_{1})$ is not assigned channel
$C_{1}$ and the link $(v_{j-1},v_{j})$ is assigned channel $C_{2}$ which is
different from $C_{out}$. We are not increasing the cost of $\P'$ by doing this.
This is depicted in Figure \ref{fig:overlap1}.

If $C_{out}=C_{2}$, we obtain $\P''_{u_{i},v_{j}}$ by going from $u_{i}$ to
$v_{j-1}$ along $\P'_{u_{i},u_{i+1}}$ and adding the links $(v_{j-1},v_{j+1})$
and $(v_{j+1},v_{j})$. Assigning channels $C_{2}$ and $C_{1}$ to links
$(v_{j-1},v_{j+1})$ and $(v_{j+1},v_{j})$ respectively, we get a path $\P''$
from $s$ to $d$. This is depicted in Figure \ref{fig:overlap2}. Now, we prove
that the cost of $\P''$ is no larger than the cost of $\P'$. Let the portion of
$\P'$ from $s$ to $v_{j-1}$ be $\P'_{s,v_{j-1}}$ and the portion from $v_{j}$ to
$d$ be $\P'_{v_{j},d}$. Applying triangle inequality in the triangles
$(v_{j-1},v_{j},v_{j+1})$ and $(v_{j},v_{j+1},u_{i+1})$, we have
$|v_{j-1}-v_{j+1}| \leq |v_{j-1}-v_{j}| + |v_{j}-v_{j+1}|$ and $|v_{j}-v_{j+1}|
\leq |v_{j+1}-u_{i+1}| + |u_{i+1}-v_{j}|$. Adding, we get,
\begin{IEEEeqnarray}{rCl}
\displaystyle
|v_{j-1}-v_{j+1}| + |v_{j}-v_{j+1}| &\leq& \left( |v_{j-1}-v_{j}| + |v_{j}-v_{j+1}| \right) + \nonumber \\
&& \left( |v_{j+1}-u_{i+1}| + |u_{i+1}-v_{j}| \right) \nonumber \\
&&
\end{IEEEeqnarray}
Adding the cost of the portions of $\P'$ from $s$ to $v_{j-1}$ and from $v_{j}$
to $d$, we get,
\begin{IEEEeqnarray}{rCl}
\displaystyle
|\P''| &\leq & |\P'_{s,v_{j-1}}| + \left( |v_{j-1}-v_{j}| + |v_{j}-v_{j+1}| \right) \nonumber \\
&& +\: \left( |v_{j+1}-u_{i+1}| + |u_{i+1}-v_{j}| \right) + d_{\P'_{v_{j},d}}
\end{IEEEeqnarray}
Since cost of the portions of $\P'$ from $v_{j+1}$ to $u_{i+1}$ and from
$u_{i+1}$ to $v_{j}$ are not less than $|v_{j+1}-u_{i+1}|$ and $|u_{i+1}-v_{j}|$
respectively, $\P''| \leq |\P'|$. The case where $\P'_{s,u_{i}}$ overlaps with
$\P'_{u_{i},u_{i+1}}$ is analogous to the above case. Instead of using the last
node of overlap, we would use the first node of overlap to obtain the shortened
path from $s$ to $d$. The number of overlaps in $\P''$ is less than that in
$\P'$ because in $\P''$, the portion of the path from $u_{i}$ to $v_{j}$ now no
longer overlaps with the portion from $v_{j}$ to $d$.
\end{proof}

{\bf Lemma \ref{lem:quadbounds}.}\ \ 
{\sl $CDCYG'_{k}$ is a CDC $t$-Spanner where
$t=\dfrac{1}{(1-2\sin{\tfrac{\theta}{2}})^{2}}$ and the number of links in
$CDCYG'_{k}$ is $O(k\cdot C^{2}\cdot |V|)$.
}
\begin{proof}
Consider an edge $(u,v)$ which is part of some CDC-path $\P$ in $\G$. Note that
we are still under the assumption that the number of nodes in $\P$ is greater
than 2. Let us consider the three cases: (i) $(u,v)$ is one of the intermediate
links other than the first and last link, (ii) $(u,v)$ is the first link in $\P$
and (iii) $(u,v)$ is the last link in $\P$. Consider the first case. If the only
channel shared between $u$ and $v$ is $C_{a}$, then $v$ has at least one more
channel, say $C_{b}$, available at it. Now, we may construct a path through
nodes which have channels $\{C_{a},C_{b}\}$. This will satisfy CDC because we
may assign channels $C_{a}$ and $C_{b}$ alternately to edges in this path. If
the number of channels shared is two or more, then, by our construction, we may
still construct a path as before between $u$ and $v$ through nodes which have
any pair of channels from the shared channels. In the remaining cases for
$(u,v)$, we may still construct paths from $u$ to $v$ because, again, at least
one of them has two or more channels. Lemmas \ref{lem:pathexists},
\ref{lem:linkspanner} and \ref{lem:untangling} still hold for these paths.
Hence, $CDCYG'_{k}$ is a CDC $t$-Spanner for $\G$. The number of links in
$CDCYG'_{k}$ at a node $v$ is $O(C^{2})$ per sector since we have three links
for each pair of channels $\{C_{a},C_{b}\}$ and the total number of pairs is
$O(C^{2})$. Hence, the total number of links in $CDCYG'_{k}$ is $O(k\cdot 
C^{2}\cdot |V|)$.

\end{proof}

\end{document}